\documentclass[11pt]{amsart}
\usepackage[english]{babel}
\usepackage{physics}
\usepackage{braket}
\usepackage{dsfont}
\usepackage{xcolor}
\usepackage{graphicx}
\usepackage{tikz}
\usetikzlibrary{matrix}

\usepackage[margin=1in]{geometry}
\usepackage[pagebackref, colorlinks = true, linkcolor = blue, urlcolor  = blue, citecolor = red]{hyperref}

\newtheorem{thm}{Theorem}[section]
\newtheorem*{thm*}{Theorem}
\newtheorem{lem}[thm]{Lemma}
\newtheorem{cor}[thm]{Corollary}

\newtheorem{defi}[thm]{Definition}
\newtheorem{prop}[thm]{Proposition}
\newtheorem{remark}[thm]{Remark}
\newtheorem{question}[thm]{Question}
\begin{document}
\author{Andreas Bluhm}
\email{andreas.bluhm@ma.tum.de}
\address{Zentrum Mathematik, Technische Universit\"at M\"unchen, Boltzmannstrasse 3, 85748 Garching, Germany}

\author{Ion Nechita}
\email{nechita@irsamc.ups-tlse.fr}
\address{Laboratoire de Physique Th\'eorique, Universit\'e de Toulouse, CNRS, UPS, France}
\title[Joint measurability and the matrix diamond]{Joint measurability of quantum effects\\and the matrix diamond}

\begin{abstract}
In this work, we investigate the joint measurability of quantum effects and connect it to the study of free spectrahedra. Free spectrahedra typically arise as matricial relaxations of linear matrix inequalities. An example of a free spectrahedron is the matrix diamond, which is a matricial relaxation of the $\ell_1$-ball. We find that joint measurability of binary POVMs is equivalent to the inclusion of the matrix diamond into the free spectrahedron defined by the effects under study.
This connection allows us to use results about inclusion constants from free spectrahedra to quantify the degree of incompatibility of quantum measurements. In particular, we completely characterize the case in which the dimension is exponential in the number of measurements. Conversely, we use techniques from quantum information theory to obtain new results on spectrahedral inclusion for the matrix diamond.
\end{abstract}

\maketitle

\tableofcontents

\section{Introduction}

One of the defining properties of quantum mechanics is the existence of incompatible observables, i.e.\ measurements that cannot be performed simultaneously \cite{Heisenberg1927, Bohr1928}.  A classic example of this behavior are the observables of position and momentum. One of the central notions to capture this property of quantum mechanics is joint measurability. Observables are jointly measurable if they arise as marginals from a common observable. This has practical implications for quantum information tasks \cite{Brunner2014}, as only incompatible observables can violate Bell inequalities \cite{Fine1982}.

It is well-known that incompatible observables can be made compatible by adding a sufficient amount of noise \cite{busch2013comparing}. Although
 many works study compatibility questions for concrete observables (see \cite{heinosaari2016} for a topical review), there has also been an interest in how much incompatibility there is in quantum mechanics and other generalized probabilistic theories \cite{busch2013comparing, Gudder2013}. In the present work, we continue this line of research by studying the degree of incompatibility in quantum mechanics in more detail. We will be interested in the compatibility regions for a fixed number of binary measurements in fixed dimension and for different types of noise.

For this, we will use tools from the study of free spectrahedra (see \cite{Helton2013} for a general introduction). Concretely, we are interested in the problem of (free) spectrahedral inclusion \cite{helton_matricial_2013}. Originally, the inclusion of free spectrahedra has been introduced as a relaxation to study the inclusion of ordinary spectrahedra \cite{ben-tal2002tractable,helton2014dilations}. In contrast to that, we will be interested in the inclusion constants for their own sake. Often, results on free spectrahedral inclusion work for large classes of spectrahedra, e.g.\ spectrahedra with symmetries \cite{helton2014dilations, davidson2016dilations}. Recently, results have been found which study maximal and minimal free spectrahedra for the $p$-norm unit balls \cite{passer2018minimal}. It is especially the latter work which is most useful to us. We would also like to mention that another point of contact between quantum information theory and free analysis is the extension (or interpolation) problem for completely positive maps, see \cite{Heinosaari2012,ambrozie2015interpolation}.

In this work, we establish the connection between free spectrahedral inclusion and joint measurability. The matricial relaxation of the $\ell_1$-ball is known as the \emph{matrix diamond} and plays a central role in our setting. We can then use results on inclusion constants for this free spectrahedron to characterize the degree of incompatibility of quantum effects in different settings. Conversely, we translate techniques to prove upper and lower bounds on quantum incompatibility to study spectrahedral inclusion. Let us note that since the problems of joint measurability and quantum steering are closely related \cite{uola2015one}, many of our results can be translated to the steering framework.

\section{Main results}

In this section, we will briefly outline the main findings of our work. Its main contribution is to connect the following two seemingly unrelated problems.

One is the problem of \emph{joint measurability of binary quantum observables}. Given a $g$-tuple of quantum effects $E_1, \ldots E_g$, we can ask the question of how much noise we have to add to the corresponding measurements fo make them jointly measurable. Joint measurability means that there exists a joint POVM $\Set{R_{i_1, \ldots, i_g}}$ from which the binary POVMs we are interested in arise as marginals. Noise can be added in different ways to a measurement. We will mainly consider the case in which we take convex combinations of a quantum measurement with a fair coin, i.e.\
\begin{equation*}
E^\prime:= s E + (1-s) I/2
\end{equation*}
for $s \in [0,1]$. The set of $g$-tuples $s \in [0,1]^g$ which make \emph{any} $g$ binary POVMs of dimension $d$ compatible will be denoted as $\Gamma(g,d)$.

The other problem comes from the field of \emph{free spectrahedra}. A free spectrahedron is a matricial relaxation of an ordinary spectrahedron. The free spectrahedron $\mathcal D_A$ is then the set of self-adjoint matrix $g$-tuples $X$ of arbitrary dimension which fulfill a given linear matrix inequality
\begin{equation*}
\sum_{i=1}^g A_i \otimes X_i \leq I.
\end{equation*}
If we only consider the scalar elements $\mathcal D_A(1)$ of this set, this is just the ordinary spectrahedron defined by the matrix tuple $A$. The inclusion problem for free spectrahedra is to find the scaling factors $s \in \mathbb R_+^g$ such that the implication
\begin{equation}\label{eq:implication}
\mathcal D_A(1) \subseteq\mathcal D_B(1) \Rightarrow s\cdot\mathcal D_A \subseteq\mathcal D_B
\end{equation}
is true. We will be interested in the case in which $\mathcal D_A$ is the matrix diamond, i.e.\ the set of matrices $X$ such that 
\begin{equation*}
\sum_{i = 1}^g \epsilon_i X_i \leq I \qquad \forall \epsilon \in \Set{-1,1}^g.
\end{equation*}
The set of all such $s$ which make the implication in Equation \eqref{eq:implication} true for any $B \in (\mathcal M_d^{sa})^g$ in this case will be written as $\Delta(g,d)$.

The main contribution of our work is to relate these two problems and use this connection to characterize $\Gamma(g,d)$. In Theorem \ref{thm:matrix-diamond-vs-effects}, we find the following:
\begin{thm*}
	Let $E \in \left(\mathcal{M}^{sa}_d\right)^g$ and let $2E -I := (2 E_1 - I_d, \ldots, 2 E_g - I_d)$. We have
	\begin{enumerate}
		\item $\mathcal{D}_{\diamond, g}(1) \subseteq\mathcal{D}_{2 E-I}(1)$ if and only if $E_1, \ldots, E_g$ are quantum effects.
		\item 	$\mathcal{D}_{\diamond, g} \subseteq\mathcal{D}_{2 E-I}$ if and only if $E_1, \ldots, E_g$ are jointly measurable quantum effects.
		\item $\mathcal{D}_{\diamond, g}(k) \subseteq\mathcal{D}_{2 E-I}(k)$ for $k \in [d]$ if an only if for any isometry $V: \mathbb C^k \hookrightarrow \mathbb C^d$, the induced compressions $V^\ast E_1 V, \ldots V^\ast E_g V$ are jointly measurable quantum effects. 
	\end{enumerate}
\end{thm*}
This shows that there is a one-to-one correspondence between different levels of the spectrahedral inclusion problem and different degrees of compatibility. Furthermore, we show in Theorem \ref{thm:jm=inclusion} that finding spectral inclusion constants corresponds to making POVMs compatible through adding noise:
\begin{thm*}
It holds that $\Gamma(g,d) = \Delta(g,d)$.
\end{thm*} 
This result allows us to use results on spectrahedral inclusion in order to characterize the set $\Gamma(g,d)$. We find that the higher dimensional generalization of the positive quarter of the unit circle plays an important role in this:
\begin{equation*}
\mathrm{QC}_g := \Set{s \in \mathbb R_+^g: \sum_{i = 1}^g s_i^2 \leq 1}.
\end{equation*}
The adaptation of some results of \cite{passer2018minimal} allows us to show in Theorem \ref{thm:LB-Passer-et-al}:
\begin{thm*}
Let $g$, $d \in \mathbb N$. Then, it holds that $\mathrm{QC}_g \subseteq \Gamma(g,d)$. In other words, for any $g$-tuple $E_1, \ldots, E_g$ of quantum effects and any positive vector $s \in \mathbb R_+^g$ with $\|s\|_2 \leq 1$, the $g$-tuple of noisy effects
$$E'_i = s_i E_i + (1-s_i) \frac{I_d}{2}$$
is jointly measurable.
\end{thm*}

If the dimension of the effects under study is exponential in the number of measurements, Theorem \ref{thm:UB-anti-commuting-F} provides us with a converse result. Again, this theorem is based on a result of \cite{passer2018minimal}.
\begin{thm*}
	Let $g \geq 2$, $d \geq 2^{\lceil (g-1)/2 \rceil}$. Then, $\Gamma(g,d)\subseteq \mathrm{QC}_g.$
\end{thm*}
Thus, for $g \geq 2$, $d \geq 2^{\lceil (g-1)/2 \rceil}$ we infer that $\Gamma(g,d) = \mathrm{QC}_g$; this equality was known previously only in the case $g=2$. However, this can no longer be the case for many measurements in low dimensions, as we point out in Section \ref{sec:discussion}. For other types of noise added to quantum measurements, we can use similar results to give upper and lower bounds on the compatibility regions. The bounds we obtain with our techniques improve greatly on past results in the quantum information literature. As an example, the best lower bound in the symmetric case came from cloning and was of order $1/g$ for fixed $g$ and large $d$, see Proposition \ref{prop:Gamma-clone-vs-Gamma-and-hat-Gamma} . Our results yield a lower bound of $1/\sqrt g$, which turns out to be exact in the regime $g \ll d$; we refer the reader to Section \ref{sec:discussion} for a detailed comparison of these bounds.

Conversely, we can use techniques from quantum information such as asymmetric cloning (Section \ref{sec:cloning}) to give bounds on spectrahedral inclusion in different settings. In particular, we introduce in this work a generalization of the notion of inclusion constants from \cite{helton2014dilations} in two directions: first, by restricting both the size and the number of the matrices appearing in the spectrahedron, and then by allowing asymmetric scalings of the spectrahedra, see Definition \ref{def:inclusion-constants}. Our contribution to the inclusion theory of free spectrahedra is going beyond the results from \cite{passer2018minimal}, by studying the asymmetric and size-dependent inclusion constants.

\section{Concepts from Quantum Information Theory}
In this section, we will start by reviewing some notions from quantum information theory related to measurements. Subsequently, we will define several versions of incompatibility of quantum measurements and show basic relations between them. For an introduction to the mathematical formalism of quantum mechanics see \cite{heinosaari2011} or \cite{watrous2018theory}, for example.

Before we move to the quantum formalism, let us introduce some basic notation. For brevity, we will write $[n] := \left\{1, \ldots, n\right\}$ for $n \in \mathbb{N}$ and $\mathbb{R}_+^g$ for $\Set{x \in \mathbb{R}^g: x_i \geq 0 ~\forall i \in [g]}$, $g \in \mathbb{N}$. Additionally, $\lceil \cdot \rceil: \mathbb R \to \mathbb Z$ will be the ceiling function. Furthermore, for $n$, $m \in \mathbb{N}$ let $\mathcal{M}_{n,m}$ be the set of complex $n \times m$ matrices. If $m = n$, we will just write $\mathcal{M}_n$. We will write $\mathcal{M}_n^\mathrm{sa}$ for the self-adjoint matrices and $\mathcal U_d$ for the unitary $d \times d$ matrices. $I_n \in \mathcal{M}_n$ will be the identity matrix. We will often drop the index if the dimension is clear from the context. For $A \in (\mathcal M_d^{sa})^g$, let $\mathcal{O}\mathcal{S}_A$ be the \emph{operator system} defined by the $g$-tuple $A$, i.e.\
\begin{equation*}
	\mathcal{O}\mathcal{S}_A := \mathrm{span}\Set{I_d, A_i : i \in [g]}.
\end{equation*}
Moreover, we will often write for such tupels $2A-I := (2A_1 - I_d, \ldots, 2A_g - I_d)$ and $V^\ast A V := (V^\ast A_1 V, \ldots, V^\ast A_g V)$ for $V \in \mathcal M_{d,k}$, $k \in \mathbb N$.

A quantum mechanical system is described by its \emph{state} $\rho \in \mathcal{S}(\mathcal{H})$, where $\mathcal{H}$ is the Hilbert space associated with the system and 
\begin{equation*}
\mathcal{S}(\mathcal{H}) := \Set{\rho \in \mathcal{B}(\mathcal{H}): \rho \geq 0, \tr[\rho] = 1}.
\end{equation*} 
In the present work, all Hilbert spaces will be finite dimensional. To describe transformations between quantum systems, we will use the concept of completely positive maps. Let $\mathcal T: \mathcal{B}(\mathcal{H}) \to \mathcal{B}(\mathcal{K})$ be a linear map with $\mathcal H$, $\mathcal K$ two Hilbert spaces. This map is $k$-positive if for $k \in \mathbb N$, the map $\mathcal T \otimes \mathrm{Id}_k :  \mathcal{B}(\mathcal{H}) \otimes \mathcal M_k \to \mathcal{B}(\mathcal{K}) \otimes \mathcal M_k$ is a positive map. A map is called completely positive if it is $k$-positive for all $k \in \mathbb N$. If this map is additionally trace preserving, it is called a \emph{quantum channel}.

Let $\mathrm{Eff}_d$ be the set of $d$-dimensional \emph{quantum effects}, i.e.\
\begin{equation*}
\mathrm{Eff}_d := \Set{E \in \mathcal{M}_d^\mathrm{sa}: 0 \leq E \leq I}.
\end{equation*}
Effect operators are useful to describe quantum mechanical measurements. In quantum information theory, measurements correspond to \emph{positive operator valued measures} (POVMs). A POVM is a set $\Set{E_i}_{i \in \Sigma}$, $E_i \in \mathrm{Eff}_d$ for all $i \in \Sigma$, such that 
\begin{equation*}
\sum_{i \in \Sigma} E_i = I.
\end{equation*}
Here, $\Sigma$ is the set of measurement outcomes, which we will assume to be finite for simplicity and equal to $[m]$ for some $m \in \mathbb{N}$. For the case of binary POVMs ($m = 2$), we will identify the POVM $\Set{E, I - E}$ with its effect operator $E \in \mathrm{Eff}_d$.

If a collection of POVMs can be written as marginals of a common POVM with more outcomes, we will say that they are jointly measurable (see \cite{heinosaari2016} for an introduction to the topic).
\begin{defi}[Jointly measurable POVMs]
We consider a collection of $d$-dimensional POVMs $\Set{E^{(i)}_j}_{j \in [m_i]}$ where $m_i \in \mathbb{N}$ for all $i \in [g]$, $g \in \mathbb{N}$. These POVMs are \emph{jointly measurable} or \emph{compatible} if there is a $d$-dimensional POVM $\Set{R_{j_1, \ldots, j_g}}$ with $j_i \in [m_i]$ such that for all $u \in [g]$ and $v \in [m_u]$,
\begin{equation*}
E^{(u)}_v =  \sum_{\substack{j_i \in [m_i] \\i \in [g]\setminus \Set{u}}} R_{j_1, \ldots ,j_{u-1}, v, j_{u+1}, \ldots, j_g}.
\end{equation*}
\end{defi}
It is well-known \cite{heinosaari2016} that not all quantum mechanical measurements are compatible. In concrete situations, the joint measurability of POVMs can be checked using a semidefinite program (SDP) (see e.g.\ \cite{Wolf2009measurements}).  Note that the SDP for $g$ binary POVMs has $2^g$ variables, so when the number of effects $g$ is large, it becomes computationally costly to decide compatibility. However, incompatible measurements can be made compatible by adding noise to the respective measurements. A trivial measurement is a POVM in which all effects are proportional to the identity. Adding noise to a measurement then means taking a convex combination of the original POVM and a trivial measurement. In order to quantify incompatibility of measurements, we can define several sets which differ in the type of noise we allow. We will restrict ourselves to binary POVMs in this work. For our first set, we allow different types of noise for every POVM:
\begin{equation*}
\Gamma^{all}(g,d) := \Set{s \in [0,1]^g: \forall E_1, \ldots E_g \in \mathrm{Eff}_d, \exists a \in [0,1]^g \mathrm{~s.t.~} s_i E_i + (1-s_i) a_i I\mathrm{~are~compatible}}.
\end{equation*}
Another possibility is to consider only balanced noise:
\begin{equation*}
\Gamma(g,d) := \Set{s \in [0,1]^g: \mathrm{~} s_i E_i + \frac{1-s_i}{2} I\mathrm{~are~compatible~}\forall E_1, \ldots E_g \in \mathrm{Eff}_d}.
\end{equation*}
Sometimes, it is inconvenient that the map from the original measurements to the ones with added noise is non-linear in the effect operators. To remedy this, we define
\begin{equation*}
\Gamma^{lin}(g,d) := \Set{s \in [0,1]^g: s_i E_i + (1-s_i) \frac{\tr[E_i]}{d} I\mathrm{~are~compatible~}\forall E_1, \ldots E_g \in \mathrm{Eff}_d}.
\end{equation*}
The restriction of this set to equal weights has appeared before in the context of quantum steering \cite{Heinosaari2015, Uola2014}.

Instead of restricting the type of noise allowed, we can also consider less general POVMs and restrict to those which are unbiased:
\begin{equation*}
\Gamma^0(g,d):= \Set{s \in [0,1]^g: s_i E_i + \frac{1-s_i}{2} I\mathrm{~are~compatible~}\forall E_1, \ldots E_g \in \mathrm{Eff}_d \mathrm{~s.t.~}\tr[E_i] = \frac{d}{2}}.
\end{equation*}
Finally, let us introduce a set of parameters related to (asymmetric) cloning of quantum states
\begin{align}
	\label{eq:def-Gamma-clone} \Gamma^{clone}(g,d) := \{s &\in [0,1]^g \, : \, \exists \mathcal T: \mathcal M_d^{\otimes g} \to \mathcal M_d \text{ unital and completely positive s.t. } \\
	\nonumber &\forall X \in \mathcal M_d, \, \forall i \in [g], \,  \mathcal T\left(I^{\otimes (i-1)} \otimes X \otimes I^{\otimes (n-i)}\right) = s_i X + (1-s_i)\frac{\tr[X]}{d} I\}.
\end{align}
All these sets are convex sets, as the next proposition shows.
\begin{prop} \label{prop:Gamma_convex}
$\Gamma^\#(g,d)$ is convex for $d$, $g \in \mathbb N$ and $\# \in \Set{all, \emptyset,lin,0,clone}$.
\end{prop}
\begin{proof}
We only prove the proposition for $\Gamma(g,d)$ here, because the proofs for the other sets are very similar. Let $s$, $t \in \Gamma(g,d)$ and $\lambda \in [0,1]$. Let further $E_1, \ldots, E_g \in \mathrm{Eff}_d$. By the choice of $s$ and $t$, we know that the $s_i E_i + (1-s_i)I/2$ and the $t_i E_i + (1-t_i)I/2$ are each jointly measurable and give rise to joint POVMs $R_{i_1, \ldots, i_g}$ and $R^\prime_{i_1, \ldots, i_g}$, respectively. Then,
\begin{equation*}
\lambda R_{i_1, \ldots, i_g} +(1-\lambda) R^\prime_{i_1, \ldots, i_g}
\end{equation*}
is again a POVM and it can easily be verified that 
\begin{align*}
&\sum_{\substack{i_j \in [2] \\j \in [g]\setminus \Set{u}}} \lambda R_{i_1, \ldots, i_{u-1}, 1,i_{u+1} ,\ldots, i_g} +(1-\lambda) R^\prime_{i_1, \ldots, i_{u-1}, 1,i_{u+1},\ldots , i_g} \\&= [\lambda s_u + (1-\lambda)t_u]E_u +[1- (\lambda s_u + (1-\lambda)t_u)]I/2.
\end{align*}
As the effects were arbitrary, this proves the assertion for $\Gamma(g,d)$.
\end{proof}
\begin{remark}\label{rmk:trivial_point}
Using convexity, it can easily be
 seen that $(1/g, \ldots, 1/g) \in \Gamma^{\#}(g,d)$, where $\# \in \Set{all, \emptyset, lin, 0, clone}$. It can be  seen that the standard basis vector $e_i$ is in each of the sets for $i \in [g]$. The above statement then follows by convexity. See also \cite{heinosaari2016} for an intuitive argument. 
\end{remark}

In the next proposition, we collect some relations between the different sets we have defined;
\begin{prop}\label{prop:inclusions-Gamma}
Let $g$, $d \in \mathbb{N}$. Then the following inclusions are true:
\begin{enumerate}
\item $\Gamma(g,d) \subseteq \Gamma^{all}(g,d)$
\item $\Gamma^{lin}(g,d) \subseteq \Gamma^{all}(g,d)$
\item $\Gamma(g,d)\subseteq \Gamma^0(g,d)$
\item $\Gamma^{lin}(g,d)\subseteq \Gamma^0(g,d)$
\item $\Gamma^{clone}(g,d) \subseteq \Gamma^{lin}(g,d)$
\item $\Gamma^{0}(g,2d) \subseteq \Gamma(g,d)$
\item $F\left(\Gamma^{all}(g,d)\right) \subseteq \Gamma(g,d)$,
where 
$$F:[0,1]^g \to [0,1]^g, \qquad F(s_1,\ldots, s_g) = \left( \frac{s_1}{2-s_1}, \ldots, \frac{s_g}{2-s_g} \right).$$
\end{enumerate}
\end{prop}
\begin{proof}
The first two assertions are true since we restrict the trivial measurements we are mixing with in both cases. The third assertion follows in the same way, but this time compatibility has to hold for less states. The fourth assertion follows since $\tr[E_i]/d = 1/2$ for the effects considered for $\Gamma^0(g,d)$.
For the fifth claim, let $s \in \Gamma^{clone}(g,d)$ be an arbitrary scaling $g$-tuple and consider quantum effects $E_1, \ldots, E_g \in \mathrm{Eff}_d$. Define, for every bit-string $b$ of length $g$  
	$$F_b := \mathcal T(E_1^{(b_1)} \otimes E_2^{(b_2)} \otimes \cdots \otimes E_g^{(b_g)}),$$
	where we set $E_i^{(1)} = E_i$ and $E_i^{(0)} = I-E_i$, and $\mathcal T$ is a map as in \eqref{eq:def-Gamma-clone}. Since the map $\mathcal T$ is (completely) positive, we have that all the operators $F_b$ are positive semidefinite. Moreover, the marginals can be computed as follows: 
	$$\sum_{b \in \{0,1\}^g, \, b_i = 1} F_b = \mathcal T(I^{\otimes i-1} \otimes E_i \otimes I^{\otimes n-i}) = s_i E_i + (1-s_i) \frac{\tr[E_i]}{d} I =: E_i',$$
	which shows that the mixed effects $E_i'$ are compatible, proving the claim. 

For the sixth assertion, let $s \in \Gamma^{0}(g,2d)$. Then, for a $g$-tuple of arbitrary $d \times d$ quantum effects $E_i$, the quantum effects (of size $2d$)
	$$s_i [E_i \oplus (I_d - E_i)] + (1-s_i) \frac{d}{2d} I_{2d}$$
are unbiased ($\tr[E_i \oplus (I_d - E_i)] = d$) and thus compatible. Truncating the above effects to their upper-left corner proves the claim. 

Let us now prove the seventh and final claim. It is enough to show that, for any effect $E \in \mathrm{Eff}_d$ and any mixture $E' = sE + (1-s)aI$ with some trivial effect $aI$ ($a \in [0,1]$), there is a further mixture $E'' = xE' + (1-x)bI = y E + (1-y) I/2$. Working out the relations between the parameters $s,x,y,a,b$, we find the following two equations
$$y = xs \quad \text{ and } \quad b = \frac{1-xs-2xa(1-s)}{2(1-x)}.$$
Asking that, for all values of $a \in [0,1]$, $b$ is also between $0$ and $1$, we obtain the desired inequality $y \leq s/(2-s)$. Let $(E^\prime_1, \ldots E_g^\prime)$ be the compatible effects corresponding to $s$. Then $E^\prime_1, \ldots, E_{j-1}^\prime, b_j I, E_{j+1}^\prime, \ldots, E_g^\prime$ are compatible as well, since we obtain a joint POVM for the effects without $E_j^\prime$ by summing over the $j$-th index and $b_j I$ is a trivial measurement and as such compatible with all effects. Then, we obtain the element $E^{\prime \prime} = (x_1 E_1^\prime + (1-x_1)b_1 I, \ldots, x_g E_g^\prime + (1-x_g)b_g I)$ from $(E_1^\prime, \ldots,  E_g^\prime)$ by successively taking convex combinations with elements of the form $(E^\prime_1, \ldots, E_{j-1}^\prime, b_j I, E_{j+1}^\prime, \ldots, E_g^\prime)$. As convex combinations of compatible tupels stay compatible (see proof of Proposition \ref{prop:Gamma_convex}), we infer that $E^{\prime \prime }$ is compatible and the assertion follows.
\end{proof}
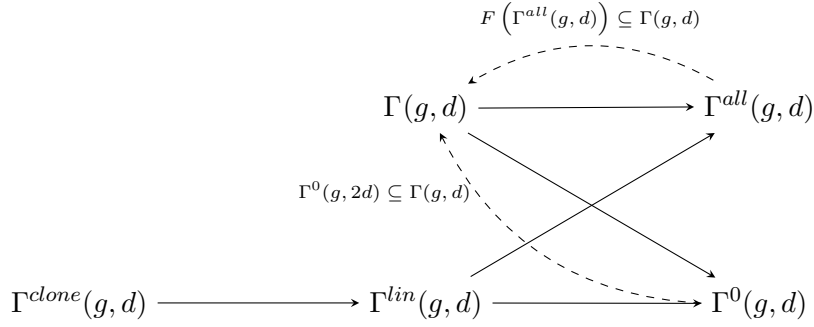
\begin{figure}
\begin{tikzpicture}
  \matrix (m) [matrix of math nodes,row sep=5em,column sep=7em,minimum width=2em]
  {
      & \Gamma^{}(g,d) & \Gamma^{all}(g,d) \\
     \Gamma^{clone}(g,d) & \Gamma^{lin}(g,d) & \Gamma^{0}(g,d) \\};
  \path[-stealth]
        (m-2-1) edge (m-2-2)
        (m-1-2) edge (m-1-3)
        (m-2-2) edge  (m-1-3)
        (m-2-2) edge (m-2-3)
        (m-1-2) edge (m-2-3)
        (m-2-3) edge [dashed, bend left] node [pos=.8,left,inner sep=2pt] {\tiny{$\Gamma^{0}(g,2d) \subseteq \Gamma(g,d)$}} (m-1-2)
        (m-1-3) edge [bend right, dashed] node [above] {\tiny{$F\left(\Gamma^{all}(g,d)\right) \subseteq \Gamma(g,d)$}}  (m-1-2);
            
\end{tikzpicture}
\caption{The different inclusions for the sets $\Gamma^\#$, $\# \in \{\emptyset, lin, all, 0, clone\}$ proven in Proposition \ref{prop:inclusions-Gamma}. Full arrows represent inclusions of sets, while dashed arrows represent special conditions.}
\label{fig:inclusions-Gamma}
\end{figure}
\begin{remark}\label{rem:Gamma-subset-Gamma-lin}
It would be interesting to see if, in general, $\Gamma(g,d') \subseteq \Gamma^{lin}(g,d)$ for some parameter $d'$ which depends on $d$. 
\end{remark}

Now we show that these sets become smaller when we increase the dimension of the effects considered.
\begin{prop}\label{prop:higherdimension}
Let $g$,$d \in \mathbb{N}$ and $\# \in \Set{all, \emptyset,lin,0,clone}$. Then
\begin{equation*}
\Gamma^{\#}(g,d+1) \subseteq\Gamma^{\#}(g,d).
\end{equation*}
\end{prop}
\begin{proof}
Let us first show the inclusion for $\Gamma^{all}$; the proofs for $\Gamma$ and $\Gamma^0$ are almost identical. Let $E_i \in \mathrm{Eff}_d$ for all $i \in [g]$. We can embed these effects into $\mathrm{Eff}_{d+1}$ by choosing $E^\prime_i = E_i \oplus 0$. Let $s \in \Gamma^{all}(g,d+1)$. Then there exists an $a \in [0,1]^g$ such that the effects
\begin{equation*}
s_i E^\prime_i +  (1-s_i) a_i I_{d+1}, \qquad i\in [g],
\end{equation*} 
are compatible. Let $V: \mathbb{C}^d \hookrightarrow \mathbb{C}^{d+1}$ be an isometry such that $V V^\ast$ is the projection onto the first $d$ entries. It is easy to check that for a POVM $\Set{R_i}_{i \in [g]}$ with $R_i \in \mathrm{Eff}_{d+1}$ for all $i \in [g]$, the set $\Set{V^\ast R_i V}_{i \in [g]}$ is again a POVM with elements in $\mathrm{Eff}_{d}$. Furthermore,
\begin{equation*}
V^\ast\left[s_i E^\prime_i +  (1-s_i) a_i I_{d+1}\right]V = s_i E_i +  (1-s_i) a_i I_{d},
\end{equation*}
which implies that the $s_i E_i +  (1-s_i) a_i I_{d}$ are compatible and therefore $s \in \Gamma^{all}(g,d)$.

To prove the claim for $\Gamma^{lin}$, we use the same idea as before, but with the following linear embedding of $d \times d$ matrices into $(d+1) \times (d+1)$ matrices
\begin{align*}
	\Psi: \mathcal M_d &\to \mathcal M_{d+1} \\
	X &\mapsto X \oplus \frac{\tr[X]}{d}.
\end{align*}
As $\tr[\Psi(X)] = (d+1)/d \tr[X]$, the claim follows.

Finally, since $\Psi$ can easily be seen to be completely positive and unital, we can use it to together with the embedding $V$ from above to define the cloning map for dimension $d$ through the one for dimension $d+1$. With $ \mathcal T_{d+1}$ and $\mathcal T_d$ the maps appearing in \eqref{eq:def-Gamma-clone} for $\Gamma^{clone}(g,d+1)$ and $\Gamma^{clone}(g,d)$, respectively, we  obtain
\begin{equation*}
\mathcal T_d(X) := V^\ast \mathcal T_{d+1}(\Psi^{\otimes g}(X))V.
\end{equation*}
It can then be verified that the map indeed has the desired properties. 
\end{proof}

\begin{remark}
We would like to point out that the sets $\Gamma(g,d)$ give rise to \emph{compatibility criteria}, i.e.~\emph{sufficient conditions for compatibility}, as follows. Let $s \in \Gamma(g,d)$ be such that $s_i > 0$ for all $i \in [g]$. Then, the following implication holds:
$$\forall i \in [g],\, \frac{1}{s_i}E_i - \frac{1-s_i}{2s_i}I_d \in \mathrm{Eff}_d \quad \implies \quad (E_1, \ldots, E_g) \text{ compatible.}$$
Indeed, using $s \in \Gamma(g,d)$ and the hypothesis, it follows that the effects
$$s_i\left[ \frac{1}{s_i}E_i - \frac{1-s_i}{2s_i}I_d \right] + (1-s_i) \frac{I_d}{2} = E_i$$
are compatible. These criteria become useful if the corresponding SDP is intractable.
\end{remark}

\section{Free spectrahedra}

In this section, we will review some concepts from the study of free spectrahedra which we will need in the rest of the paper. We will start with the definition of free spectrahedra and their inclusion. Then, we will review the link between spectrahedral inclusion and positivity properties of certain maps. All the theory needed in this work can be found in \cite{helton_matricial_2013,helton2014dilations,davidson2016dilations}.

Let $A \in \left(\mathcal{M}_d^{sa}\right)^g$. The \emph{free spectrahedron at level $n$} corresponding to this $g$-tuple of matrices is the set
\begin{equation*}
\mathcal{D}_A(n) := \Set{X \in \left(\mathcal{M}_n^{sa}\right)^g: \sum_{i=1}^g A_i \otimes X_i \leq I_{nd}}.
\end{equation*}
For $n = 1$, this is a usual spectrahedron defined by a \emph{linear matrix inequality}. The \emph{free spectrahedron} is then defined as the (disjoint) union of all these sets, i.e.\
\begin{equation*}
\mathcal{D}_A := \bigcup_{n \in \mathbb{N}} \mathcal{D}_A(n).
\end{equation*}
A free spectrahedron which we will very often encounter is the \emph{matrix diamond of size $g$}. It is defined as 
\begin{equation*}
\mathcal{D}_{\diamond, g}(n) = \Set{X \in \left(\mathcal{M}_n^{sa}\right)^g: \sum_{i=1}^g \epsilon_i X_i \leq I_n~\forall \epsilon \in \Set{-1,+1}^g}.
\end{equation*}
To see that this is a free spectrahedron, we can take the direct sum of all these constraints. The matrices defining this free spectrahedron are thus diagonal. At level 1, the matrix diamond is the unit ball of the $\ell_1$-norm. Therefore, it is obviously bounded.
For free spectrahedra, the inclusion $\mathcal{D}_A \subseteq\mathcal{D}_B$ means $\mathcal{D}_A(n) \subseteq\mathcal{D}_B(n)$ for all $n \in \mathbb{N}$.
Inclusion at the level of spectrahedra ($n = 1$) does not guarantee inclusion of the free spectrahedra. That is, the implication
\begin{equation*}
\mathcal{D}_A(1) \subseteq\mathcal{D}_B(1) \Rightarrow \mathcal{D}_A \subseteq\mathcal{D}_B
\end{equation*}
does not hold in general. However, scaling the set $\mathcal{D}_A$ down, the implication becomes eventually true. 
\begin{defi}\label{def:inclusion-constants}
Let $\mathcal{D}_A$ be the free spectrahedron defined by $A \in \left(\mathcal{M}_d^{sa}\right)^g$. The inclusion set $\Delta_{\mathcal{D}_A}(k)$ is defined as 
\begin{equation*}
\Delta_{\mathcal{D}_A}(k) := \Set{s \in \mathbb{R}_+^g: \forall B \in \left(\mathcal{M}_k^{sa}\right)^g~\mathcal{D}_A(1) \subseteq\mathcal{D}_B(1) \Rightarrow s \cdot \mathcal{D}_A \subseteq\mathcal{D}_B}.
\end{equation*}
For $\mathcal{D}_A = \mathcal{D}_{\diamond,g}$, we will write $\Delta(g,k)$ for brevity. Here, the set
\begin{equation*}
s \cdot \mathcal{D}_A := \Set{(s_1 X_1, \ldots, s_g X_g) :  X \in \mathcal{D}_A}
\end{equation*} is the \emph{(asymetrically) scaled} free spectrahedron.
\end{defi}
Note that the definition above generalizes the one from \cite{helton2014dilations} by restricting the size of the matrices defining the containing spectrahedra and by allowing non-symmetric scaling; one recovers the definition of inclusion constants from \cite{helton2014dilations} by considering the largest constant $s \geq 0$ such that 
$$ (\underbrace{s, \ldots, s}_{g \text{ times}}) \in \Delta_{\mathcal D_A}(k), \qquad \forall k \geq 1.$$

As in the case of POVMs, we are also interested in the inclusion constant set where we restrict to inclusion into free spectrahedra defined by traceless matrices.
\begin{defi}
	Let $\mathcal D_A$ be a spectrahedron defined by $A \in \left(\mathcal{M}_d^{sa}\right)^g$. The \emph{traceless-restricted} inclusion set $\Delta^0_{\mathcal D_A}(d)$ of $\mathcal D_A$ is defined as
\begin{align*} \Delta^0_{\mathcal D_A}(k) := \{s \in \mathbb R_+^g \, : \, &\forall C \in (\mathcal{M}^{sa}_k)^g \text{ s.t. }  \forall i \in [g], \, \tr[C_i] = 0,\\
&\qquad \mathcal D_A(1) \subseteq \mathcal D_C(1) \implies s\cdot \mathcal D_A \subseteq \mathcal D_C \}.
\end{align*}
For $\mathcal{D}_A = \mathcal{D}_{\diamond,g}$, we will again write $\Delta^0(g,k)$ for brevity.
\end{defi}
The next proposition shows that both inclusion sets we have defined are convex.
\begin{prop}
Let $A \in \left(\mathcal{M}_d^{sa}\right)^g$. Both $\Delta_{\mathcal D_A}(k)$ and $\Delta^0_{\mathcal D_A}(k)$ are convex.
\end{prop}
\begin{proof}
Let $B \in \left(\mathcal{M}_k^{sa}\right)^g$ and $X \in \mathcal D_A$. Further, let $s$, $t \in \Delta_{\mathcal D_A}(k)$ and $\lambda \in [0,1]$. The assumptions on $s$, $t$ yield
\begin{equation*}
\sum_{i = 1}^g B_i \otimes (\lambda s_i + (1-\lambda)t_i)X = \lambda \sum_{i = 1}^g B_i \otimes s_i X + (1-\lambda) \sum_{j = 1}^g B_j \otimes t_j X \leq I.
\end{equation*}
This proves the first assertion, because $B$ was arbitrary. The second assertion follows in a very similar manner. 
\end{proof}

The inclusion of spectrahedra can be related to positivity properties of a certain map. Let $A \in (\mathcal M_D^{sa})^g$ and $B \in (\mathcal M_d^{sa})^g$ define the free spectrahedra $\mathcal D_A$ and $\mathcal D_B$, respectively. Let $\Phi:\mathcal{O}\mathcal{S}_A \to \mathcal M_d$ be the unital map defined as
\begin{equation*}
\Phi: A_i \mapsto B_i \qquad \forall i \in [g].
\end{equation*}
Then, we can find a one-to-one relation between properties of $\Phi$ and the inclusion of the free spectrahedra at different levels. This has been proven in \cite[Theorem 3.5]{helton_matricial_2013} for real spectrahedra and we include a proof in the complex case for convenience. 
\begin{lem} \label{lem:cpspectra}
Let $A \in (\mathcal M_D^{sa})^g$ and $B \in (\mathcal M_d^{sa})^g$. Further, let $\mathcal D_A(1)$ be bounded. Then $\mathcal D_A(n) \subseteq\mathcal D_B(n)$ holds if and only if $\Phi$ as given above is $n$-positive. In particular, $\mathcal D_A \subseteq\mathcal D_B$ if and only if $\Phi$ is completely positive. 
\end{lem}
\begin{proof}
The ``only if'' direction is true by the unitality and $n$-positivity of $\Phi$. For the ``if'' direction, let $Y \in \mathcal M_n^{sa}(\mathcal O \mathcal S_A)$. Without loss of generality, we can assume $I_D, A_1, \ldots A_g$ to be linearly independent. Then 
\begin{equation*}
Y = I_D \otimes X_0 - \sum_{i = 1}^g A_i \otimes X_i
\end{equation*}
for $(X_0, \ldots X_g) \in \mathcal M_n^g$. We claim that $X_0, \ldots X_g$ are self-adjoint.
Then $(I_D \otimes e_i^\ast) (Y-Y^\ast)( I_D \otimes e_j) = 0$ for all $i$, $j \in [n]$ and an orthonormal basis $\Set{e_i}_{i = 1}^n$ of $\mathbb C^n$ if and only if $\langle e_i, (X_l - X_l^\ast)e_j \rangle$ for all $i$, $j \in [n]$ and for all $l \in [g] \cup \Set{0}$. This proves the claim. If $Y \geq 0$, it holds that $X_0 \geq 0$. Let us assume that this is not the case. Then there exists an $x \in \mathbb C^{n}$ such that $\langle x, X_0 x \rangle < 0$. Positivity of $Y$ yields 
\begin{equation*}
- \sum_{i=1}^g \langle x, X_i x \rangle A_i > 0.
\end{equation*}
Therefore, $\lambda(\langle x, X_1 x \rangle, \ldots, \langle x, X_g x \rangle) \in \mathcal D_A(1)$ for all $\lambda \geq 0$. This contradicts the assumption that $\mathcal D_A(1)$ is bounded. Let us now assume that $Y \geq 0$ and that $X_0 > 0$. Then $(\Phi \otimes Id_n) Y \geq 0$, because $X_0^{-1/2} X X_0^{-1/2} \in \mathcal D_A(n) \subseteq\mathcal D_B(n)$. For $Y \geq 0 $ and $X_0 \geq 0$, positivity of $(\Phi \otimes Id_n) Y$ follows from exchanging $X_0$ by $X_0 + \epsilon I_n$, $\epsilon > 0$ and letting $\epsilon$ go to zero.
\end{proof}
\begin{remark} The complete positivity of $\Phi$ can be checked using an SDP \cite{helton_matricial_2013, Heinosaari2012}. Therefore, the inclusion problem at the level of free spectrahedra is efficiently solvable. This is not necessarily the case for the usual spectrahedra at level 1, because checking the positivity of a linear map is in general a hard problem (the set of positive maps between matrix algebras is dual to the set of separable states, and deciding weak membership into the latter set is known to be NP-hard \cite{gurvits2003classical}). Seeing the free spectrahedral inclusion problem as a \emph{relaxation} of the corresponding problem for (level 1) spectrahedra is a very useful idea in optimization, see \cite{ben-tal2002tractable,helton_matricial_2013}.
\end{remark}
Using the previous lemma, we obtain a useful corollary if we assume that $\mathcal D_A(1)$ is bounded, which is enough for us.  
\begin{cor}\label{cor:disenough}
Let $A \in (\mathcal M_D^{sa})^g$ and $B \in (\mathcal M_d^{sa})^g$. Moreover, let $\mathcal D_A(1)$ be bounded. Then $\mathcal D_A(d) \subseteq\mathcal D_B(d)$ if and only if $\mathcal D_A \subseteq\mathcal D_B$.
\end{cor}
\begin{proof}
From Lemma \ref{lem:cpspectra}, $\mathcal D_A(d) \subseteq\mathcal D_B(d)$ is equivalent to $\Phi$ being $d$-positive. Since $\Phi$ maps to $\mathcal M_d$, this is equivalent to complete positivity of the map \cite[Theorem 6.1]{Paulsen2002}. The claim then follows by another application of Lemma \ref{lem:cpspectra}.
\end{proof}
\begin{remark}\label{rmk:bounded}
The result of Corollary \ref{cor:disenough} without the boundedness assumption has appeared before in \cite[Lemma 2.3]{helton2014dilations} for real spectrahedra with a longer proof. Their proof carries over to the complex setting. Therefore, the boundedness assumption is not necessary, but it shortens the proof considerably.
\end{remark}

\section{Spectrahedral inclusion and joint measurability}
In this section, we establish the link between joint measurability of effects and the inclusion of free spectrahedra. The main result of this work is Theorem \ref{thm:matrix-diamond-vs-effects} which we prove at the end of this section. It connects the inclusion of the matrix diamond into a spectrahedron with the joint measurability of the quantum effects defining this spectrahedron. Before we can prove the theorem, we will need two lemmas concerning compressed versions of a (free) spectrahedron. 

\begin{lem} \label{lem:isometrylarger}
Let $k \in \mathbb N$, $1 \leq k \leq d$ and let $V: \mathbb C^k \hookrightarrow \mathbb C^d$ be an isometry. For $A \in \left(\mathcal M^{sa}_d \right)^g$, it holds that $\mathcal D_A \subseteq\mathcal D_{V^\ast A V}$.
\end{lem} 
\begin{proof}
Let $X \in \mathcal D_A(n)$. Then from the definition
\begin{equation*}
\sum_{i = 1}^g A_i \otimes X_i \leq I_{dn}.
\end{equation*}
Multiplying the equation by $V \otimes I_n$ from the right and by its adjoint from the left, it follows that
\begin{equation*}
\sum_{i=1}^g V^\ast A_i V \otimes X_i \leq I_{kn}.
\end{equation*}
Here, we have used that $V$ is an isometry and that the map $Y \mapsto W^\ast Y W$ for matrices $Y$, $W$ of appropriate dimensions is completely positive.
\end{proof}

\begin{lem} \label{lem:intersection}
Let $k \in \mathbb N$, $1 \leq k \leq d$ and let $V: \mathbb C^k \hookrightarrow \mathbb C^d$ be an isometry. For $A \in \left(\mathcal M^{sa}_d \right)^g$, it holds that
\begin{equation} \label{eq:isometrysets}
\mathcal D_A(k) = \bigcap_{V: \mathbb C^k \hookrightarrow \mathbb C^d~\mathrm{isometry}} \mathcal D_{V^\ast A V}(k).
\end{equation}
\end{lem}
\begin{proof}
Let us denote the right hand side of Eq. \eqref{eq:isometrysets} by $\mathcal C$. From Lemma \ref{lem:isometrylarger}, it follows that $\mathcal D_A(k) \subseteq\mathcal C$. For the reverse inclusion, let $X \in \mathcal C$. This implies especially $X \in \left( \mathcal M_k^{sa} \right)^g$. We write
\begin{equation*}
Y := \sum_{i = 1}^g A_i \otimes X_i - I_{dk}.
\end{equation*}
To prove the assertion, we need to show that $Y \geq 0$. Let $y \in \mathbb C^d \otimes \mathbb C^k$ be a unit vector with Schmidt decomposition
\begin{equation*}
y = \sum_{i = 1}^k \sqrt{\lambda_i} e_i \otimes f_i,
\end{equation*}
where $\Set{e_i}_{i = 1}^k$ and $\Set{f_j}_{j = 1}^k$ are orthonormal families in $\mathbb C^d$ and $\mathbb C^k$, respectively. Further, $\lambda_i \geq 0$ and such that $\sum_{i = 1}^k \lambda_i = 1$. We then have 
\begin{equation*}
\langle y, Y y \rangle = \sum_{i,j = 1}^k \sqrt{\lambda_i \lambda_j} \langle (e_i \otimes f_i), Y (e_j \otimes f_j)\rangle.
\end{equation*}
Let $\Omega = \sum_{i = 1}^k g_i \otimes g_i$ be an unnormalized maximally entangled state with an orthonormal basis $\Set{g_i}_{i = 1}^k$ of $\mathbb C^k$. Moreover, let $V: \mathbb C^k \to \mathbb C^d$ and $Q: \mathbb C^k \to \mathbb C^k$ be defined as
\begin{equation*}
V = \sum_{i = 1}^k e_i g_i^\ast, \qquad Q = \sum_{j = 1}^k \sqrt{\lambda_{j}}f_j g_j^\ast.
\end{equation*}
Then, $V: \mathbb C^k \hookrightarrow \mathbb C^d$ is an isometry. Therefore, $(V^\ast \otimes I_k)Y(V \otimes I_k) \geq 0$ by assumption, as $X \in \mathcal C$. Hence, 
\begin{equation*}
\langle y, Y y \rangle = \langle \Omega , (I_k \otimes Q^\ast)(V^\ast \otimes I_k)Y(V \otimes I_k)(I_k \otimes Q)\Omega \rangle \geq 0.
\end{equation*}
Thus, $Y$ is positive semidefinite, because $y$ was arbitrary. 
\end{proof}

\begin{thm}\label{thm:matrix-diamond-vs-effects}
	Let $E \in \left(\mathcal{M}^{sa}_d\right)^g$ and let $2E -I := (2 E_1 - I_d, \ldots, 2 E_g - I_d)$. We have
	\begin{enumerate}
		\item $\mathcal{D}_{\diamond, g}(1) \subseteq\mathcal{D}_{2 E-I}(1)$ if and only if $E_1, \ldots, E_g$ are quantum effects.
		\item 	$\mathcal{D}_{\diamond, g} \subseteq\mathcal{D}_{2 E-I}$ if and only if $E_1, \ldots, E_g$ are jointly measurable quantum effects.
		\item $\mathcal{D}_{\diamond, g}(k) \subseteq\mathcal{D}_{2 E-I}(k)$ for $k \in [d]$ if an only if for any isometry $V: \mathbb C^k \hookrightarrow \mathbb C^d$, the induced compressions $V^\ast E_1 V, \ldots V^\ast E_g V$ are jointly measurable quantum effects. 
	\end{enumerate}
\end{thm}
\begin{proof}
	Let us start with the first point. Since $\mathcal{D}_{\diamond, g}(1)$ is a convex polytope, we need to check inclusion only at extreme points. That means that the first assertion holds if and only if $\pm e_i \in \mathcal D_{2E-I}$ for all $i \in [g]$, where $\Set{e_i}_{i = 1}^g$ is the standard basis in $\mathbb R^g$. We have 
	\begin{align*}
		e_i \in \mathcal D_E(1) &\iff 2E_i - I \leq I \iff E_i \leq I\\
		-e_i \in \mathcal D_E(1) &\iff -(2E_i-I)\leq I \iff E_i \geq 0,
	\end{align*} 
	proving the first claim.

	We now characterize the free spectrahedral inclusion from the second point. In the following, we will identify diagonal matrices with vectors and the subalgebra of diagonal $2^g \times 2^g$-matrices with $\mathbb C^{2^g}$. The operator system associated with $\mathcal D_{\diamond, g}$ 
	 is 
$$\mathbb C^{2^g} \supseteq \mathcal{OS}_{\diamond, g} := \mathrm{span}\{v_0, v_1, \ldots, v_g\},$$
where, indexing the $2^g$ coordinates by sign vectors $\varepsilon \in \{\pm 1\}^g$,
\begin{align*}
v_0(\varepsilon) &= 1\\
v_i(\varepsilon) &= \varepsilon(i) \qquad \forall i \in [g].
\end{align*} 
Here, $\varepsilon(i)$ is the $i$-th entry of the vector $\varepsilon$. The dimension of this operator system is $g+1$. We define a map $\Phi: \mathcal{O}\mathcal{S}_{\diamond, g} \to \mathcal{M}_d$ by
	\begin{align*}
	v_0  &\mapsto I \\
	v_i & \mapsto 2 E_i - I \qquad \forall i \in [g].
	\end{align*}
The spectrahedral inclusion $\mathcal{D}_{\diamond, g} \subseteq\mathcal{D}_{2 E-I}$ holds if and only if the map $\Phi$ is completely positive (Lemma \ref{lem:cpspectra}). If this is the case, Arveson's extension theorem (see \cite[Theorem 6.2]{Paulsen2002} for a finite-dimensional version) guarantees the existence of a completely positive extension $\tilde{\Phi}$ of this map to the whole algebra $\mathbb C^{2^g}$, because $\mathcal{O}{S}_{\diamond, g}$ is an operator system. As $\mathbb C^{2^g}$ is a commutative matrix algebra, it is enough to show that the extension is positive \cite[Theorem 3.11]{Paulsen2002}.  To find such an extension $\tilde\Phi: \mathbb C^{2^g} \to  \mathcal M_d$, we consider the basis $(g_\eta)_{\eta \in \{\pm 1\}^g}$ of the vector space $\mathbb C^{2^g}$ which is defined as follows:
$$g_\eta(\varepsilon) = \mathbf 1_{\varepsilon = \eta} \geq 0.$$
Here, $\mathbf 1_{\varepsilon = \eta} = 1$ if $\varepsilon = \eta$ and zero otherwise. Let us write $G_\eta := \tilde \Phi(g_\eta)$; since the $g_\eta$ are positive, the (complete) positivity of $\tilde \Phi$ is equivalent to $G_\eta \geq 0$, for all $\eta$. 

We have, for all $\varepsilon \in \{\pm 1\}^g$,
$$1=v_0(\varepsilon) = \sum_{\eta} \mathbf 1_{\varepsilon = \eta} = \sum_\eta g_\eta(\epsilon),$$
and thus we can rewrite
$$\tilde \Phi(v_0) = I \iff \sum_\eta G_\eta = I.$$

We also have
$$v_i(\varepsilon) = \varepsilon(i) = 2\mathbf 1_{\varepsilon(i) = +1} - 1 = 2\sum_{\eta \, : \, \eta(i) = +1} \mathbf 1_{\varepsilon = \eta} - 1 = 2 \sum_{\eta \, : \, \eta(i) = +1} g_\eta(\epsilon) - \sum_\eta g_\eta(\epsilon),$$
and thus we have, for all $i \in [g]$, 
$$\tilde \Phi(v_i) = 2E_i -I \iff \sum_{\eta \, : \, \eta(i) = +1} G_\eta = E_i.$$

Collecting all these facts, we have shown that the map $\Phi$ extends to a (completely) positive map on the whole $\mathbb C^{2^g}$ if and only if there exist operators $(G_\eta)_{\eta \in \{\pm 1\}^g}$ such that
\begin{align*}
\forall \eta \in \{\pm 1\}^g, \quad G_\eta &\geq 0\\
\sum_\eta G_\eta &= I\\
\forall i \in [g], \quad \sum_{\eta \, : \, \eta(i) = +1} G_\eta &= E_i;
\end{align*}
but these are precisely the conditions for the joint measurability of the effects $E_1, \ldots, E_g$ and we are done with the second point.
For the third assertion, it follows from Lemma \ref{lem:intersection} that $\mathcal{D}_{\diamond, g}(k) \subseteq\mathcal{D}_{2 E-I}(k)$ holds if and only if $\mathcal{D}_{\diamond, g}(k) \subseteq\mathcal{D}_{2 V^\ast E V-I}(k)$ for any isometry $V: \mathbb C^k \hookrightarrow \mathbb C^d$. Further, Corollary \ref{cor:disenough} asserts that this is equivalent to $\mathcal{D}_{\diamond, g} \subseteq\mathcal{D}_{2 V^\ast E V-I}$ for all isometries $V$ as above. The claim then follows from the second assertion of this theorem.
\end{proof}

\begin{remark}\label{rem:compression-jm-is-jm}
	The fact that the second point of the theorem above implies the third point, read on the quantum effects side of the equivalence, is a well known fact\footnote{We thank Teiko Heinosaari for pointing this out to us.}: compressions of jointly measurable effects are jointly measurable. 
\end{remark}

\begin{remark}If $E \in (\mathcal M_d^{sa})^g$ is a $g$-tuple of pairwise commuting matrices, then $\mathcal D_{\diamond, g}(1) \subseteq \mathcal D_{2E-I}(1)$ implies $\mathcal D_{\diamond, g} \subseteq \mathcal D_{2E-I}$. This is true, because the effects on the right hand side generate a commutative matrix subalgebra. Inclusion at level one then implies that the corresponding map $\Phi$ is positive. As its range is contained in a commutative matrix algebra, this also implies that $\Phi$ is completely positive \cite[Theorem 3.9]{Paulsen2002}, which then yields the inclusion at the level of free spectrahedra. This recovers the well-known result from quantum information theory that pairwise commuting effects are jointly measurable.
\end{remark}

\begin{remark}We can recover another result from quantum information theory, namely that effects of the form $a I$, $a \in [0,1]$, are trivially compatible with any effect. This corresponds to the fact that 
\begin{equation}\label{eq:trivial_extension}
\mathcal D_{\diamond, g} \subseteq \mathcal D_{2E-I} \iff \mathcal D_{\diamond, g+1} \subseteq \mathcal D_{(2E-I, \alpha I)}
\end{equation}
for any $E \in (\mathcal M_d^{sa})^g$. Here, we write $\alpha = 2a-1$, i.e.\ $\alpha \in [-1,1]$. This can best be seen at the level of maps. It is easy to see that the $v_i$ defining $\mathcal D_{\diamond, g}$ (c.f. proof of Theorem \ref{thm:matrix-diamond-vs-effects}) can be written as 
\begin{equation*}
v_i = \prod_{j=1}^{i-1} I_2 \otimes \mathrm{diag}[+1,-1] \otimes \prod_{j=i+1}^{g} I_2. 
\end{equation*}
Let $\Phi_g$ be the map corresponding to the left hand side and $\Phi_{g+1}$ be the one corresponding to the right hand side of Equation \eqref{eq:trivial_extension}. For the `` only if''-direction, we can simply define $\Phi_{g}(A) = \Phi_{g+1}(A \otimes I_2)$, where $A \in \mathbb C^{2^g}$. For the ``if''-direction, we define the linear map $\Psi: \mathbb C^2 \to \mathbb C$ as
\begin{equation*}
\Psi: (1,0) \mapsto \frac{\alpha+1}{2}, \qquad \Psi: (0,1) \mapsto \frac{1-\alpha}{2}.
\end{equation*}
This map is unital, positive and therefore also completely positive. We can then set $\Phi_{g+1} = \Phi_g \otimes \Psi$. It can then be checked using the above expression for the $v_i$ that this map has indeed the desired properties.
\end{remark}

\begin{thm} \label{thm:jm=inclusion}
It holds that $\Gamma(g,d) = \Delta(g,d)$ and that $\Gamma^0(g,d) = \Delta^0(g,d)$.
\end{thm}
\begin{proof}
Let $s \in \mathbb R^g$. Then $s \in \Gamma(g,d)$ if and only if $s_{1}E_1 + (1-s_{1})I/2, \ldots, s_{g}E_g + (1-s_{g})I/2$ are jointly measurable for all effects $E_1, \ldots E_g \in \mathrm{Eff}_d$. It can easily be seen that  
\begin{equation*}
 s\cdot\mathcal D_{\diamond,g} \subseteq\mathcal D_{2E-I} \iff \mathcal D_{\diamond,g} \subseteq\mathcal D_{s\cdot(2E-I)} \iff \mathcal D_{\diamond,g} \subseteq\mathcal D_{2E^\prime-I},
\end{equation*}
where $E_i^\prime = s_i E_i + (1-s_{i})I/2$ for all $i \in [g]$. Therefore, Theorem \ref{thm:matrix-diamond-vs-effects} implies that $s \in \Gamma(g,d)$ is equivalent to the implication  $\mathcal D_{\diamond,g}(1) \subseteq \mathcal D_{2E-I}(1) \Rightarrow s\cdot\mathcal D_{\diamond,g} \subseteq\mathcal D_{2E-I}$ for all $E \in (\mathcal M_d^{sa})^g$. This is equivalent to $s \in \Delta(g,d)$, because $X \mapsto 2X-I$ is a bijection on $\mathcal M_d^{sa}$. The second assertion follows since $X \mapsto 2X-I$ is a bijection between $\Set{X \in \mathcal M_d^{sa}:\tr[X] = d/2}$ and $\Set{A \in \mathcal M_d^{sa}:\tr[A] = 0}$.
\end{proof}
\section{Lower bounds from cloning} \label{sec:cloning}

In this section we provide, using known facts about the set $\Gamma^{clone}$, lower bounds for $\Gamma^{lin}$ and $\Gamma$. We start by recalling the main results from the theory of \emph{symmetric} cloning. 

\begin{thm}[{\cite[Theorem 7.2.1]{Keyl2002}}]
	For a quantum channel $\mathcal T: \mathcal M_d^{\otimes N} \to \mathcal M_d^{\otimes g}$, consider the quantities
	\begin{equation*}
		\mathcal F_{c,1}(\mathcal T) = \inf_{j \in [g]} \inf_{\sigma \mathrm{~pure}} \tr[\sigma^{(j)}\mathcal T(\sigma^{\otimes N})]
	\end{equation*}
	where $\sigma^{(j)} = I^{\otimes (j-1)} \otimes \sigma \otimes I^{\otimes (g - j)} \in \mathcal M_d^{\otimes g}$ and
	\begin{equation*}
		\mathcal F_{c, all}= \inf_{\sigma \mathrm{~pure}} \tr[\sigma^{\otimes g}\mathcal T(\sigma^{\otimes N})].
	\end{equation*}
	These quantities are both maximized by the \emph{optimal quantum cloner}
	\begin{equation*}
		\mathcal T(\rho) = \frac{d[N]}{d[g]} S_g(\rho \otimes I)S_g.
	\end{equation*}
	Here, $d[g] = \binom{d + g - 1}{g}$ is the dimension of the symmetric subspace $\vee^g \mathbb C^d \subseteq (\mathbb C^d)^{\otimes g}$ and $S_g$ is the corresponding orthogonal projection.
\end{thm} 
From \cite[Equation 3.7]{Werner1998}, we know further that
\begin{equation*}
	\tilde{\rho} := \mathrm{tr}_{i^c}[\mathcal T(\rho)] = \gamma \rho + (1 - \gamma)I/d,
\end{equation*}
where $\gamma = (g+d)/(g(1 + d))$ (for $N = 1$). Here, $\mathrm{tr}_{i^c}$ means tracing out all systems but the $i$-th one. Going to the dual picture, we can compute that for $E \in \mathcal M_d$
\begin{equation*}
	\tr[\tilde{\rho}E] = \tr[\rho\left(\gamma E + (1 - \gamma)\frac{\tr E }{d}I\right)].
\end{equation*}
We can therefore identify $E^\prime = \gamma E + (1 -\gamma) I_E$, where $I_E$ is the trivial effect $\tr{E}/d I$ depending on $E$. Therefore, $\gamma$ is a lower bound on the joint measurability of a family of effects, 
\begin{equation*}
	\mathcal T^\ast\left(\Set{E^{(i)}_1}_i \otimes \ldots \otimes \Set{E^{(j)}_g}_j\right)
\end{equation*}
being the joint observables.

Inserting the expression for $\gamma$ from the symmetric cloning bounds and using Proposition \ref{prop:inclusions-Gamma}, we obtain the following result; note that below, the second quantity is always larger than the third one. 
\begin{prop}\label{prop:Gamma-clone-vs-Gamma-and-hat-Gamma}
	For all $g,d \geq 2$, 
	\begin{align*}
		\frac{g+d}{g(1+d)} (\underbrace{1, 1, \ldots, 1}_{g \text{ times}}) &\in \Gamma^{clone}(g,d) \subseteq \Gamma^{lin}(g,d)\subseteq \Gamma^{all}(g,d)\\
\frac{g+2d}{g(1+2d)} (\underbrace{1, 1, \ldots, 1}_{g \text{ times}}) &\in \Gamma^{clone}(g,2d) \subseteq \Gamma^{0}(g,2d)\subseteq \Gamma(g,d)\subseteq \Gamma^{all}(g,d)\\
		\frac{g+d}{g + d(2g-1)} (\underbrace{1, 1, \ldots, 1}_{g \text{ times}}) &\in F(\Gamma^{clone}(g,d)) \subseteq F(\Gamma^{all}(g,d)) \subseteq \Gamma(g,d).
	\end{align*}
\end{prop}

In the general, non-symmetric case, the exact form of the set $\Gamma^{clone}(g,d)$ has been computed, by different methods, in \cite{kay2016optimal} and \cite{studzinski2014group}; the following restatement of the optimal cloning probabilities is taken from the former reference. 

\begin{thm}{\cite[Theorem 1, Section 2.3]{kay2016optimal}}\label{thm:Gamma-clone-asymmetric}
	For any $g,d \geq 2$, 
	\begin{align}\label{eq:optimal-asymmetric-cloning}
	\Gamma^{clone}(g,d)  =& \Bigg\{s \in [0,1]^g \, : \, (g+d-1)\left[g-d^2+d + (d^2-1)\sum_{i=1}^g s_i\right] \\&\leq \left(\sum_{i=1}^g \sqrt{s_i(d^2-1)+1} \right)^2\Bigg\}. \nonumber
	\end{align}
	Using the variables $t_i := s_i(d^2-1)+1 \in [1,d^2]$, we have the simpler expression
	$$\Gamma^{clone}(g,d)  = \left\{s \in [0,1]^g \, : \, \|t\|_1 - \frac{\|t\|_{1/2}}{g+d-1} \leq d(d-1)\right\},$$
	where $\|\cdot\|_p$ denotes the $\ell_p$-quantity on $\mathbb R^g$: $\|t\|_p = (\sum_{i=1}^g |t_i|^p)^{1/p}$.
\end{thm}
\begin{proof}
	The formula is exactly \cite[Equation (5)]{kay2016optimal}, after the change of variables $F_i = s_i + (1-s_i)/d$, for all $i \in [g]$.
\end{proof}

\begin{remark}
	Note that the symmetric cloning optimal probability is recovered by setting $s_1=s_2=\cdots = s_g$ in the result above, yielding the the maximal value
	$$s_{\max} = \frac{g+d}{g(d+1)}.$$
\end{remark}

\begin{remark}\label{rem:Gamma-clone-d-infty}
	In the regime $d \to \infty$, the left hand side of \eqref{eq:optimal-asymmetric-cloning} behaves like $d^3(\|s\|_1 - 1)$, whereas the right hand side behaves like $d^2 \|s\|_{1/2}$.
	Hence, asymptotically, the achievable cloning probabilities should satisfy $\sum_i s_i \leq 1$; the set such values is the probability simplex, i.e.~the convex hull of the points $\{e_i\}_{i=1}^g$, where $e_i$ is the basis vector having a 1 in position $i$ and zeros elsewhere. 
\end{remark}

We discuss next the special cases of pairs and triplets, i.e.~$g=2,3$. The most studied for (asymmetric) cloning is the $g=2$ case (see, e.g. \cite{cerf1998asymmetric} or the more recent \cite[Theorem 3]{hashagen2016universal}). We plot in the left panel of Figure \ref{fig:clone-2-3-d} the sets $\Gamma^{clone}(2,d)$ for various values of $d$, as subsets of $\Gamma(2,d) = \mathrm{QC}_2$. 
\begin{prop}
	For all $d \geq 2$, we have 
	$$\Gamma^{clone}(2,d) = \{(s,t) \in [0,1]^2 \, : \, s+t-\frac 2 d \sqrt{(1-s)(1-t)} \leq 1\}.$$
\end{prop}
\begin{proof}
	To see that the condition above is equivalent to equation \eqref{eq:optimal-asymmetric-cloning} for $g=2$, one can solve both for $t$ and show that the answer is the following:
	$$t \leq 1-s-\frac{2(1-s)}{d^2}+\frac{2 \sqrt{1-s+(d^2-1)s(1-s)}}{d^2}.$$
\end{proof}

The case $g=3$ is also worth mentioning, since one can obtain manageable expressions for the set $\Gamma^{clone}(3,d)$. In the right panel of Figure \ref{fig:clone-2-3-d}, we plot the slice $\Gamma^{clone}(3,d) \cap \{(s,t,t)\}$ for various values of $d$ (this corresponds to asking that the ``quality'' of the second and third clones are identical), against the Euclidean ball (see Section \ref{sec:LB-QC} for the relevance of the quarter-circle). 

\begin{prop}
	For all $d \geq 2$, we have, either in explicit or in parametric form \cite{iblisdir2005multipartite}
	\begin{align*}
		\Gamma^{clone}(3,d) &= \left\{(s,t,u) \in [0,1]^3 \, : \, (d+2)\left[3-d^2+d+(d^2-1)(s+t+u)\right]\leq \right.\\
		&\left. \qquad\qquad\qquad \left(\sqrt{(d^2-1)s+1} + \sqrt{(d^2-1)t+1} + \sqrt{(d^2-1)u+1}\right)^2\right\}\\
		&= \left\{ \left(1-b^2-c^2-\frac{2 b c}{d+1},1-c^2-a^2-\frac{2 c a}{d+1},1-a^2-b^2-\frac{2 a b}{d+1} \right) \, : \right. \\
		&\left. \qquad\qquad\qquad a^2+b^2+c^2 + 2(ab+bc+ca)/d \leq 1,~a,b,c\geq 0\right\}.
	\end{align*}
\end{prop}

\begin{figure}[htbp]
	\begin{center}
		\includegraphics[scale=.6]{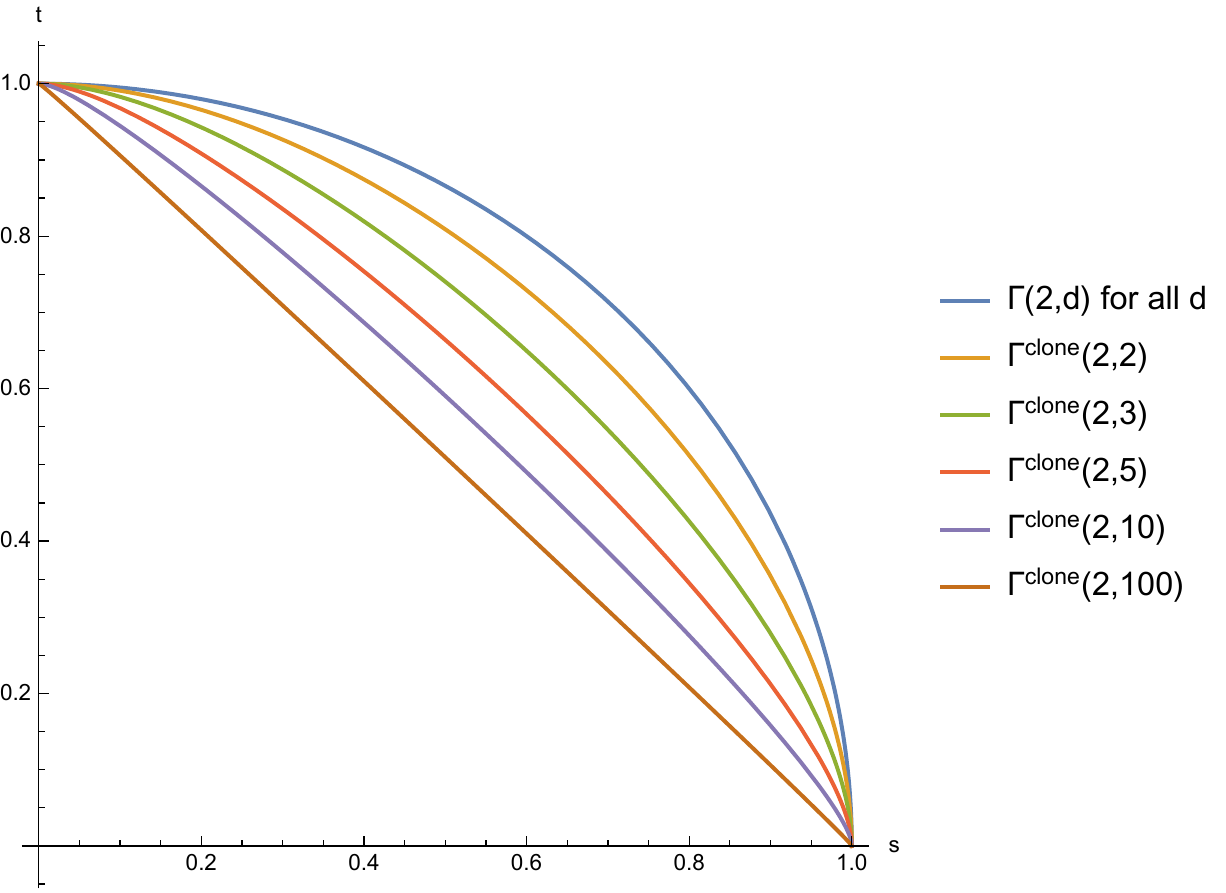}  \includegraphics[scale=.6]{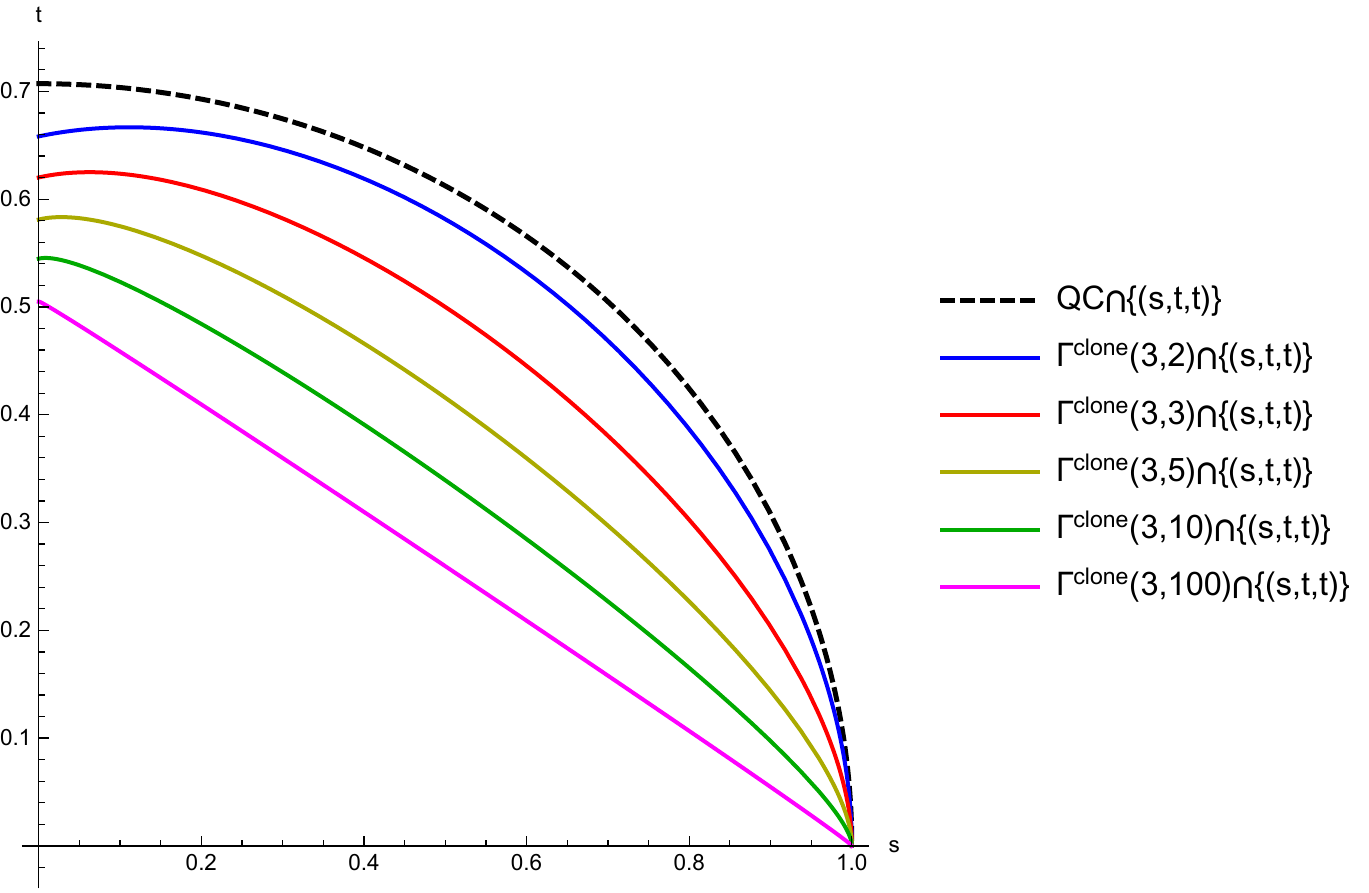}
		\caption{Left: the sets $\Gamma^{clone}(2,d)$ for $d=2,3,5,10,100$ as subsets of $\Gamma(2,d)$, which is a quarter-circle for all $d \geq 2$ (see Corollary \ref{cor:large-d-Gamma-is-QC}). Right: the cuts $\Gamma^{clone}(3,d) \cap \{(s,t,t)\}$ for $d=2,3,5,10,100$.}
		\label{fig:clone-2-3-d}
	\end{center}
\end{figure}

\section{Lower bounds from free spectrahedra}
\subsection{Dimension dependent and symmetric lower bounds}

This part basically reproduces the proof of Theorem 1.4 in \cite{helton2014dilations} and shows that making minor changes, the proof also works in the case where the spectrahedra are given in terms of complex instead of real matrices. Note that in this case, we obtain an inclusion constant of $2d$ instead of merely $d$. Let us first recall a lemma from \cite{helton2014dilations}, which was proved there for real matrices but carries over without change.

\begin{lem}[Lemma 8.2 from {\cite{helton2014dilations}}]\label{lem:sums}
Suppose $T = (T_{j,l})$ is a $k \times k$ block matrix with blocks of equal size. If $\left\|{T_{j,l}}\right\|_\infty \leq 1$ for every $j$, $l \in [k]$, then $\|{T}\|_\infty \leq k$.
\end{lem}

\begin{prop} \label{prop:heltonsymmetric}
Let $A \in \left (\mathcal M_D^{sa} \right)^g$ and $B \in \left (\mathcal M_d^{sa} \right)^g$. Suppose further that $-\mathcal D_A \subseteq\mathcal D_A$ and that $\mathcal D_A(1)$ is bounded. If $\mathcal D_A(1) \subseteq\mathcal D_B(1)$, then $\mathcal D_A\subseteq2d\mathcal D_B$.
\end{prop}
\begin{proof}
Fix some level $n$ and consider $\{e_l\}_{l=1}^n$, the standard orthonormal basis for $\mathbb C^n$. Fix $1 \leq s \neq t \leq n$ and set $p_{s,t}^{\pm} = 1/\sqrt{2}(e_s \pm e_t) \in \mathbb C^n$, $\phi_{s,t}^\pm = 1/\sqrt{2}(e_s \pm \mathrm{i} e_t) \in \mathbb C^n$. Further, let
\begin{equation*}
P_{s,t}^\pm = I_D \otimes p_{s,t}^\pm, \qquad \Phi_{s,t}^\pm = I_D \otimes \phi_{s,t}^\pm.
\end{equation*}
Then $\left(P_{s,t}^\pm\right)^\ast P_{s,t}^\pm = I_{D} = \left(\Phi_{s,t}^\pm\right)^\ast \Phi_{s,t}^\pm$. Let $M \in \mathcal M_n^{sa}$, $C \in \mathcal M_D$. Then 
\begin{align*}
\left(P_{s,t}^+\right)^\ast (C \otimes M) P_{s,t}^+ - \left(P_{s,t}^-\right)^\ast (C \otimes M) P_{s,t}^- &= 2 C \otimes \mathrm{Re}(M)_{s,t} \\
\left(\Phi_{s,t}^+\right)^\ast (C \otimes M) \Phi_{s,t}^+ - \left(\Phi_{s,t}^-\right)^\ast (C \otimes M) \Phi_{s,t}^- & = -2 C \otimes \mathrm{Im}(M)_{s,t}.
\end{align*}
Let $X \in \mathcal D_A(n)$ and let $Z = \sum_j A_j \otimes X_j$. By hypothesis, $-X \in \mathcal D_A(n)$ as well, thus $\pm Z \leq I_{Dn}$. By the above calculations, also $\pm \left(P_{s,t}^\pm\right)^\ast Z P_{s,t}^\pm \leq I_D$. Therefore, $\pm \left(p_{s,t}^\pm\right)^\ast X p_{s,t}^\pm \in \mathcal D_A(1)$ for all $s \neq t$. Convexity of $\mathcal D_A(1)$ together with the above implies $\pm \mathrm{Re}(X)_{s,t} = \pm (\mathrm{Re}(X_1)_{s,t}, \ldots, \mathrm{Re}(X_g)_{s,t}) \in \mathcal D_A(1)$. The same holds true for $s = t$ if one chooses $p_{t} = e_t$ and makes the necessary adjustments in the above argument. Considering $\phi_{s,t}^\pm$, we find that $\pm \mathrm{Im}(X)_{s,t} \in \mathcal D_A(1)$ for all $s$, $t \in [n]$ as well. Now set 
\begin{equation*}
T_{s,t} = \sum_j B_j \otimes (X_j)_{s,t}.
\end{equation*}
It holds that
\begin{equation*}
\left\|{T_{s,t}}\right\|_{{\infty}} \leq \left\|{\sum_j B_j \otimes \mathrm{Re}(X_j)_{s,t}}\right\|_{{\infty}} + \left\|{\sum_k B_k \otimes \mathrm{Im}(X_k)_{s,t}}\right\|_{{\infty}}.
\end{equation*}
Moreover, we know that
\begin{equation*}
- I_n \leq \sum_j B_j \otimes \mathrm{Re}(X_j)_{s,t} \leq I_n \qquad - I_n \leq \sum_k B_k \otimes \mathrm{Im}(X_k)_{s,t} \leq I_n
\end{equation*}
as the real and imaginary parts of the entries of $X$ have been found to be in $\mathcal D_A(1)$ and therefore also in $\mathcal D_B(1)$ by hypothesis. Combining the two findings, it follows that $\left\|{T_{s,t}}\right\|_{{\infty}} \leq 2$. An application of Lemma \ref{lem:sums} to $T/2$ allows us to conclude that $\left\|{T}\right\|_{{\infty}} \leq 2n$. Thus, 
\begin{equation*}
- I_{dn} \leq \frac{1}{2n}\sum_j B_j \otimes X_j \leq I_{dn}
\end{equation*}
which implies $\pm 1/(2n) X \in \mathcal D_B(n)$. At level $n=d$, this implies $\mathcal D_A(d) \subseteq2d\mathcal D_B(d)$. 
Since $B \in (\mathcal M_d^{sa})^g$, an application of Corollary \ref{cor:disenough} 
proves the assertion.
\end{proof}
\begin{remark}
The assumption that $\mathcal D_A(1)$ is bounded is not necessary and does not appear in \cite{helton2014dilations}, see Remark \ref{rmk:bounded}. 
\end{remark} 

Exploiting the link between inclusion of free spectrahedra and joint measurability, the previous proposition corresponds to:
\begin{cor} \label{cor:heltonspoint}
Let $s = (1/(2d), \ldots 1/(2d))$. Then $s \in \Gamma(g,d)$.
\end{cor}
\begin{proof}
The matrix diamond is symmetric, i.e.\ it holds that $-\mathcal D_{\diamond, g} = \mathcal D_{\diamond, g}$. Further, $\mathcal D_{\diamond, g}(1)$ is bounded. From Proposition \ref{prop:heltonsymmetric}, it follows that $\mathcal D_{\diamond, g}(1) \subseteq\mathcal D_A(1)$ implies $\mathcal D_{\diamond, g}\subseteq 2d\mathcal D_A$ for any $A \in \left(\mathcal M_d^{sa}\right)$. Thus, $s \in \Delta(g,d)$. The claim then follows from Theorem \ref{thm:jm=inclusion}.
\end{proof}

\begin{remark}
In \cite{helton2014dilations}, Proposition \ref{prop:heltonsymmetric} was proven for spectrahedra defined by real matrices and with $d$ instead of $2d$. We point out that this result cannot hold in the complex case. Consider $d = 2$, $g = 3$, $A_i = \sigma_i$ and $B_i = \sigma_i^\top$ for $i \in [3]$. Here, $\sigma_i$ are the usual Pauli matrices. In this case, the operator system spanned by the $A_i$ is the whole matrix algebra. Let $s_0 = (s, s, s)$ with $s \in [0,1]$. For $s_0 \mathcal D_A \subseteq\mathcal D_B$ to hold, the map
\begin{equation*}
\Phi^\prime(A) = s A^\top + (1-s) \tr[A]/2 I
\end{equation*}
for $A \in \mathcal M_2$ must be completely positive. This follows from Lemma \ref{lem:cpspectra} (see also \cite[Section 1.4]{helton2014dilations}). A short calculation shows that the minimal eigenvalue of the Choi matrix \cite[Section 2.2.2]{watrous2018theory} of this map is $-s/2 + (1-s)/4$. For the Choi matrix to be positive, we thus require $s \leq 1/3$, which is fulfilled by $s = 1/2d$ but not by $s = 1/d$. However, this calculation does not imply that $2d$ is optimal in Proposition \ref{prop:heltonsymmetric}, leaving this question open.
\end{remark}

\subsection{Dimension independent lower bounds}\label{sec:LB-QC}

We restate in this section one implication of \cite[Theorem 6.6]{passer2018minimal}, which, interpreted in terms of inclusion constants, yields Theorem \ref{thm:LB-Passer-et-al}. For the convenience of the reader and for the sake of being self-contained, we reproduce the proof with several simplifications and written in a language more familiar to quantum information specialists. 

Let $\mathrm{MatBall}_g$ be the \emph{matrix ball} (see \cite[Chapter 7]{pisier2003introduction}, \cite[Section 14]{helton2014dilations}, \cite[Section 9]{davidson2016dilations} for the different operator structures one can put on the $\ell_2$-ball)
$$\mathrm{MatBall}_g:=\{(X_1, \ldots, X_g) \in (\mathcal M_d^{sa})^g \, : \, \sum_{i=1}^g X_i^2 \leq I\}.$$

We recall the following result \cite[Lemma 6.5]{passer2018minimal}.
\begin{lem}\label{lem:diamond-in-ball}
For all $g \geq 1$, $\mathcal D_{\diamond, g} \subseteq \mathrm{MatBall}_g$.
\end{lem}
\begin{proof}
Let $(X_1,\ldots, X_g) \in \mathcal D_{\diamond, g}$. By definition, we have that, for all $\varepsilon \in \{\pm 1\}^g$,
$$\sum_{i=1}^g \varepsilon_i X_i \leq I.$$
Squaring the relation above, we get
$$\sum_{i} X_i^2 + \sum_{i \neq j} \varepsilon_i \varepsilon_j X_iX_j \leq I.$$
Averaging the above inequality for all values of $\varepsilon$, we are left with $\sum_i X_i^2 \leq I$, which is the claim we aimed for.
\end{proof}

Define the ``quarter-circle'' 
\begin{equation}\label{eq:def-QC}
\mathrm{QC}_g := \Set{s \in \mathbb R_+^g : s_1^2 + \ldots + s_g^2 \leq 1}
\end{equation}
to be part of the unit disk contained in the positive orthant.

\begin{thm}{\cite[Theorem 6.6]{passer2018minimal}}\label{thm:LB-Passer-et-al}
Let $A \in (\mathcal M^{sa}_d)^{g}$. For any vector $s \in \mathbb R_+^g$ such that $\sum_i s_i^2 \leq 1$ and any spectrahedron $\mathcal D_A$, whenever $\mathcal D_{\diamond, g}(1) \subseteq \mathcal D_A(1)$, we have $s \cdot \mathrm{MatBall}_g \subseteq \mathcal D_A$. In particular, $s \cdot \mathcal D_{\diamond, g} \subseteq \mathcal D_A$. In terms of inclusion constants, we have
$$\forall g,d,\qquad \mathrm{QC}_g \subseteq \Delta(g,d).$$
\end{thm}
\begin{proof}
Using Lemma \ref{lem:diamond-in-ball}, under the hypotheses, we only need to show the inclusion  $s \mathrm{MatBall}_g \subseteq \mathcal D_A$. To this end, consider a $g$-tuple of $n \times n$ self-adjoint matrices $(X_1, \ldots, X_g)$ such that $\sum_{i=1}^g X_i^2 \leq I$. We claim that this inequality implies that, for all $s$ as in the statement, 
$$\sum_{i=1}^g s_i |X_i| \leq I.$$
Indeed, this follows from the general matrix inequality 
$$\left\| \sum_i B_i  C_i \right\|_{\infty} \leq \left\|\sum_i B_i B_i^*\right\|_\infty^{1/2} \left\|\sum_i C_i^*C_i\right\|_\infty^{1/2}.$$

The above inequality can be seen to hold by writing the $B_i$ in the first row of a larger matrix and the $C_i$ in the first column of another such matrix. Writing now $X_i = Y_i - Z_i$ with positive semidefinite operators $Y_i,Z_i$ in such a way that $|X_i| = Y_i + Z_i$, we also have
$$\sum_{i=1}^g s_i Y_i + \sum_{i=1}^g s_i Z_i \leq I.$$
We interpret the last inequality as $\{s_i Y_i\}_{i=1}^g \sqcup \{s_i Z_i\}_{i=1}^g$ being a partial POVM, and we apply the Naimark dilation theorem (see \cite[Section 2.2.8]{nielsen2010quantum} or \cite[Theorem 2.42]{watrous2018theory}). Hence, there exist an isometry $V: \mathbb C^n \to \mathbb C^n \otimes \mathbb C^{2g+1}$ and $2g$ mutually orthogonal projections $P_i,Q_i \in \mathcal M_{n(2g+1)}$ such that $s_iY_i = V^*P_iV$ and $s_iZ_i = V^*Q_iV$. We thus have $s_i X_i = V^*(P_i-Q_i)V$, and the operators $R_i:=P_i-Q_i$ are commuting, normal and with joint spectrum in $\mathcal D_{\diamond, g}(1)$. Thus, with $E$ the joint spectral measure of $R$,
\begin{align*}
\sum_{i = 1}^g A_i \otimes s_{i} X_i &= \sum_{i = 1}^g A_i \otimes V^\ast R_i V \\
&= \int_{\mathcal D_{A}(1)} \left(\sum_{i = 1}^g A_i y_i\right)\otimes V^\ast dE(y)V\\
&\leq I\otimes I.
\end{align*}
This shows that $s\mathrm{MatBall_g}(k) \subseteq\mathcal D_A(k)$ and the assertion follows as $k$ was arbitrary.
\end{proof}

\begin{remark} Corollary 7.17 of \cite{davidson2016dilations} shows that $(1/g, \ldots 1/g) \in \Delta_{\mathcal D_A}$ for all $A \in  ( \mathcal M_d^{sa})^g$ such that $\mathcal D_A(1)$ is invariant under projection onto some orthonormal basis. This result thus holds in particular for the matrix diamond, but also for more general spectrahedra. It corresponds to the observation of Remark \ref{rmk:trivial_point} that $(1/g, \ldots 1/g) \in \Gamma(g,d)$. In the concrete situation that $\mathcal D_A = \mathcal D_{\diamond,g}$, the statement of Theorem \ref{thm:LB-Passer-et-al} is much stronger, as one might expect.
\end{remark}

\section{Upper bounds}
We present in this section two upper bounds (i.e.~containing sets) for the $\Gamma$ and $\Gamma^0$ sets, one coming from quantum information theory \cite{zhu2015information} and another one coming from matrix convex set theory \cite{passer2018minimal}. These two upper bounds are interesting in two different regimes: the first one applies when the number of POVMs is larger than the dimension of the quantum system, while the second one applies in the complementary regime, where the dimension is large with respect to the number of POVMs. Another important difference between the two results below is that the first one (Theorem \ref{thm:UB-Zhu}) deals with the set $\Gamma^{lin}$, while the second one (Theorem \ref{thm:UB-anti-commuting-F}) deals with the set $\Gamma^0$.

\subsection{Zhu's necessary condition for joint measurability}\label{sec:zhu}

We start by recalling Zhu's incompatibility criterion from \cite{zhu2015information}; see also \cite{zhu2016universal} for the mathematical details. To do so, define for a non-zero operator $A \in \mathcal M_d$, 
$$\mathcal G(A) := \frac{|A \rangle\langle A|}{\tr[ A]} \in \mathcal M_{d^2}^{sa},$$
where $|A \rangle \in \mathbb C^{d^2}$ is the vectorization of the matrix $A$. In the same vein, if $A^\circ := A - \tr[A]/d I$ denotes the traceless version of $A$, let 
$$\overline{\mathcal G}(A) :=  \frac{|A^\circ \rangle\langle A^\circ|}{\tr[A]} \in \mathcal M_{d^2}^{sa}.$$
We also extend additively the definitions above to POVMs
$$\mathcal G^\#(\{E_i\}) := \sum_i \mathcal G^\#(E_i),$$
where $\mathcal G^\#$ denotes either $\mathcal G$ or $\overline{\mathcal G}$. Using the remarkable fact that the functions $\mathcal G^\#$ are subadditive, Zhu has showed the following result in \cite[equations (10,11)]{zhu2015information}. 

\begin{prop}\label{prop:Zhu-criterion}
	If a set of $g$ POVMs $\{E^{(1)}\}, \ldots, \{E^{(g)}\}$ on $\mathcal M_d$ are compatible, then 
	$$\min\{\tr[H] \, : \, H \geq \mathcal G(\{E^{(i)}\}) ,\, \forall i\in[g]\} = 1+\min\{\tr[H] \, : \, H \geq \overline{\mathcal G}(\{E^{(i)}\}) ,\, \forall i\in[g]\} \leq d.$$
\end{prop}

It turns out that the semidefinite program appearing in the result above is particularly easy in the case where the $d^2 \times d^2$ matrices $\overline{\mathcal G}(\{E^{(1)}\}), \ldots, \overline{\mathcal G}(\{E^{(g)}\})$ have orthogonal supports; if that is the case, then the optimal $H$ is the sum of the matrices $\overline{\mathcal G}(\{E^{(i)}\})$, and the condition above reads 
$$\sum_{i=1}^g \tr[\overline{\mathcal G}(\{E^{(i)}\})] \leq d-1.$$
In order to exploit this phenomenon, let $G_{max}(d)$ be the maximal integer $g$ such that there exist $E_1,\ldots,E_g$ non-trivial orthogonal projections in $\mathcal M_d$ with the property $d\tr[E_iE_j] = (\tr[E_i])(\tr[E_j])$ for all $1\leq  i <j \leq g$.

\begin{thm}\label{thm:UB-Zhu}
	For all dimensions $d$ and all $1 \leq g \leq G_{max}(d)$, we have
	$$\Gamma^{lin}(g,d) \subseteq \sqrt{d-1}\mathrm{QC}_g.$$
\end{thm}
\begin{proof}
	Let $g,d$ be as in the statement, and consider the 2-outcome POVMs $\{E_i, I_d - E_i\}$, $i \in [g]$, where $E_i$ are such that $d\tr[E_iE_j] = (\tr[ E_i])(\tr[E_j])$ for all $1\leq  i <j \leq g$. Since the effects $E_i$ are non-trivial orthogonal projections, the previous condition is equivalent to $\tr[E_i^\circ E_j^\circ] = 0$, for all $i \neq j$. Fix $s \in \Gamma^{lin}(g,d)$; from the definition of the set $\Gamma^{lin}(g,d)$, it follows that the effects $s_i E_i + (1-s_i)d^{-1}\tr[ E_i] I$ are compatible, and thus, by Proposition \ref{prop:Zhu-criterion}, 
	\begin{equation}\label{eq:Zhu-upper-bound-lin}
	\sum_{i=1}^g \tr[\overline{\mathcal G}(\{s_iE_i + (1-s_i)d^{-1}\tr[ E_i] I, s_i(I_d-E_i) + (1-s_i)d^{-1}(d-\tr[E_i]) I\})] \leq d-1.
	\end{equation}
	Let us compute, for fixed $i$, the general term in the sum above. Start by computing
	$$\overline{\mathcal G}(\{E_i, I_d - E_i\}) = \frac{|E_i^\circ \rangle\langle E_i^\circ|}{\tr[E_i]} + \frac{|(I_d - E_i)^\circ \rangle\langle (I_d - E_i)^\circ|}{d-\tr[ E_i]} = \frac{d|E_i^\circ \rangle\langle E_i^\circ|}{\tr[E_i](d-\tr[ E_i])}.$$
	Using that $E_i$ is a (non-trivial) projection, we get
	$$\tr[\overline{\mathcal G}(\{E_i, I_d - E_i\})]  = \frac{d\tr[(E_i^\circ)^2]}{\tr[ E_i](d-\tr[E_i])} = \frac{d[\tr[E_i^2] - \tr[E_i]^2/d]}{\tr[E_i](d-\tr[E_i])}=1.$$
	For the noisy version, mixing with the identity does not change the trace, hence
\begin{align*}
\overline{\mathcal G}(\{s_iE_i + (1-s_i)d^{-1}\tr[E_i] I,& s_i(I_d-E_i) + (1-s_i)d^{-1}(d-\tr[E_i]) I\}) \\
&=\frac{d s_i^2 |E_i^\circ \rangle\langle E_i^\circ|}{(\tr[E_i])(d-\tr[E_i])}= s_i^2 \overline{\mathcal G}(\{E_i, I_d - E_i\}).
\end{align*}
and thus, taking the trace
$$\tr[\overline{\mathcal G}(\{s_iE_i + (1-s_i)d^{-1}\tr[E_i] I, s_i(I_d-E_i) + (1-s_i)d^{-1}(d-\tr[E_i]) I\})]= s_i^2.$$
	Zhu's condition \eqref{eq:Zhu-upper-bound-lin} implies thus
	$$\sum_{i=1}^g s_i^2 \leq d-1,$$
	proving the claim. 
\end{proof}

\begin{remark}
An analysis of the proof above shows that the same result holds for the set $\Gamma^0$ instead of $\Gamma^{lin}$, but with an extra restriction on the operators $E_i$: we must ask that $\tr[E_i] = d/2$ for all $i \in [g]$. We leave the existence of large tuples of such operators as an open problem.
\end{remark}

Let us now discuss the function $G_{max}(d)$. First, note that in order for the upper bound in the result above to be non-trivial, we must have $G_{max}(d) \geq d$. Below, we give two lower bounds on the function $G_{max}$, conditional on the existence of \emph{mutually unbiased bases} (MUBs) and \emph{symmetric informationally complete POVMs} (SIC-POVMs).

Recall that $k$ orthonormal bases $\{x^{(i)}_j\}_{j=1}^d$, $i \in [k]$ are called \emph{mutually unbiased} if and only if for all $i_1 \neq i_2$ and all $j_1, j_2$, $|\langle x^{(i_1)}_{j_1} , x^{(i_2)}_{j_2} \rangle| = 1/\sqrt d$. The maximal number of MUBs in $\mathbb C^d$ is $d+1$, and this bound is attained if $d$ is a prime power; very few other existence results are known, see \cite{durt2010mutually} for a review. If there exist $k$ MUBs in dimension $d$, then $G_{max}(d) \geq k$. This follows by setting $E_i = \sum_{j \in J_i} |x^{(i)}_j \rangle \langle x^{(i)}_j|$, where $\{x^{(i)}_j\}_{j=1}^d$ is the $i$-th MUB and $J_i$ is some non-trivial subset of $[d]$. We record next the following consequence of Theorem \ref{thm:UB-Zhu}.
\begin{cor}\label{cor:Zhu-UB-Gamma-lin-MUB}
	If $d=p^n$ is a prime power, then, for all $g \leq d+1$,
$$\Gamma^{lin}(g,d) \subseteq \sqrt{d-1}\mathrm{QC}_g.$$
\end{cor}

Recall that a unit rank POVM $\{d^{-1} |x_i \rangle \langle x_i|\}_{i=1}^{d^2}$ is called \emph{symmetric and informationally complete} if and only if for all $i \neq j$, $|\langle x_i , x_j \rangle|^2 = 1/(d+1)$. Whether a SIC-POVM exists in dimension $d$ is a challenging question and an ongoing research subject. Analytic examples of SIC-POVMs have been constructed for $d = 1, \ldots, 21, 24, 28, 30, 31, 35, 37, 39, 43, 48, 124, 323$, and numerical constructions exist for much larger values of $d$, see \cite{scott2010symmetric} for a review and \cite{appleby2018constructing} for some recent progress. The use of SIC-POVMs yields the following corollary of Theorem \ref{thm:UB-Zhu}.
\begin{cor}\label{cor:Zhu-UB-Gamma-lin-SIC-POVM}
	If there exists a SIC-POVM $\{x_i\}_{i=1}^{d^2}$ in dimension $d$, then $G_{max}(d+1) \geq d^2$ and thus, for all $g \leq d^2$, we have
$$\Gamma^{lin}(g,d+1) \subseteq \sqrt{d}\mathrm{QC}_g.$$
\end{cor}
\begin{proof}
Let $y_i = (\sqrt{t},\sqrt{1-t}x_i)$ for $i \in [d^2]$, where $t\in [0,1]$ is a parameter. Using the fact that the $x_i$ are a SIC-POVM, we have that $|\langle y_i , y_j\rangle|^2 = t+ (1-t)/(d+1)$, so for $t= 1/d^2$, $|\langle y_i , y_j\rangle|^2 = 1/d$. Setting $E_i = |y_i \rangle \langle y_i|$ then proves the claim.
\end{proof}

\subsection{Pairwise anti-commuting unitary operators and spectrahedral inclusion constants}

We present now a different type of upper bound, this time on the set $\Delta^0 = \Gamma^0$. What follows is based on \cite[Theorem 6.6]{passer2018minimal}. We adapt the proof there to our setting by taking into account the system dimension $d$. The main ingredient of the construction in \cite{passer2018minimal} is the following Hurwitz-Radon like result.

\begin{lem}{\cite{newman1932note} or \cite[Theorem 1]{hrubes2016families}}\label{lem:Hurwitz-Radon}
	For $d = 2^k$, $k \in \mathbb N_0$, there exist $2k+1$ anti-commuting, self-adjoint, unitary matrices $F_1,\ldots, F_{2k+1} \in \mathcal U_d$. Moreover, $2^k$ is the smallest dimension where such a $(2k+1)$-tuple exists. 
\end{lem}
A $(2k+1)$-tuple as above is sometimes called a \emph{spin system} in operator theory, see, e.g.~\cite{pisier2003introduction}. One can easily construct such matrices recursively, as follows.  
For $k=0$, simply take $F^{(0)}_1 := [1]$. For $k \geq 1$, define 
$$F^{(k+1)}_i = \sigma_X \otimes F^{(k)}_i~\forall i \in [2k+1] \quad \text{ and } \quad F^{(k+1)}_{2k+2} = \sigma_Y \otimes I_{2^{k}}, \quad F^{(k+1)}_{2k+3} = \sigma_Z \otimes I_{2^{k}}, $$
where $\sigma_{X,Y,Z}$ are the Pauli matrices
$$\sigma_X = \begin{bmatrix}
0 & 1 \\ 1 & 0
\end{bmatrix}, \quad \sigma_Y = \begin{bmatrix}
0 & -\text{i} \\ \text{i} & 0
\end{bmatrix} \quad \text{ and } \quad \sigma_Z = \begin{bmatrix}
1 & 0 \\ 0 & -1
\end{bmatrix}.$$
For example, we have $F^{(1)}_1 = \sigma_X$, $F^{(1)}_2 = \sigma_Y$, $F^{(1)}_3 = \sigma_Z$, and 
$$F^{(2)}_1 = \sigma_X \otimes \sigma_X, \qquad F^{(2)}_2 = \sigma_X \otimes \sigma_Y, \qquad F^{(2)}_3 = \sigma_X \otimes \sigma_Z, \qquad F^{(2)}_4 = \sigma_Y \otimes I_2, \qquad F^{(2)}_5 = \sigma_Z \otimes I_2.$$
\begin{remark}Note that our construction differs from the one in \cite{passer2018minimal}, because we aim for the smallest dimension which contains $g$ anti-commuting, self-adjoint and unitary elements. This way, we obtain $d \geq 2^{\lceil (g-1)/2 \rceil}$ instead of $d \geq 2^{g-1}$ in the next theorem.
\end{remark}

\begin{thm}\label{thm:UB-anti-commuting-F}
	Let $g \geq 2$, $d \geq 2^{\lceil(g-1)/2\rceil}$ and consider $s \in \mathbb R_+^g$ such that for any spectrahedron $\mathcal D_A$ defined by traceless matrices $A_i \in \mathcal M_d$,  $\mathcal D_{\diamond,g}(1) \subseteq \mathcal D_A(1)$ implies $s \cdot\mathcal D_{\diamond,g} \subseteq \mathcal D_A$. 
	Then, $\sum_i s_i^2 \leq 1$. In terms of inclusion constants, we have
	$$\forall d\geq 2^{\lceil(g-1)/2 \rceil}, \qquad \Delta(g,d) \subseteq \Delta^0(g,d) \subseteq \mathrm{QC}_g.$$
\end{thm}
\begin{proof}
	Let us consider $g$ anti-commuting, self-adjoint, unitary matrices $F_1,\ldots, F_g \in \mathcal U_d$ as in the construction following Lemma \ref{lem:Hurwitz-Radon}; these matrices also enjoy the property of being traceless when $g\geq 2$. Let $\mathcal D_{\overline{F}}$ be the spectrahedron defined by the matrices $\overline{F_i}$, where $\overline{F_i}$ is the entry-wise complex conjugate of $F_i$. Since the matrices $\overline{F_i}$ are unitary, it is clear that $\mathcal D_{\diamond,g}(1) \subseteq \mathcal D_{\overline{F}}(1)$. 
	
	Assume now that $s {\cdot} \mathcal D_{\diamond,g} \subseteq \mathcal D_{\overline{F}}$ for some non-negative $g$-tuple $s$. Put $\hat s:=s/\|s\|_2$. We claim that $(\hat s_i F_i)_{i=1}^g \in \mathcal D_{\diamond,g}$. Indeed, for any choice of signs $\varepsilon_i$, we have
	$$\left\|\sum_{i=1}^g \varepsilon_i \hat s_i F_i\right\|_\infty = \left\| \left(\sum_{i=1}^g \varepsilon_i \hat s_i F_i\right)^2\right\|_\infty^{1/2} = \left\| \sum_{i=1}^g \hat s_i^2I  + \sum_{i \neq j} \varepsilon_i \varepsilon_j \hat s_i \hat s_j F_iF_j\right\|_\infty^{1/2} = \left\| \sum_{i=1}^g \hat s_i^2I \right\|_\infty^{1/2} = 1.$$
	In the equality above, we have used the fact that the cross terms in the sum obtained by expanding the square vanish; it is this behavior of the matrices $F_i$ that renders them useful in operator theory. From the hypothesis, it follows that $(s_i\hat s_i F_i)_{i=1}^g \in \mathcal{D}_{\overline{F}}$; in particular, we have
	$$\left\|\sum_{i=1}^g \frac{s_i^2}{\|s\|_2} \overline{F_i} \otimes F_i\right\|_\infty = \frac{\sum_{i=1}^g s_i^2}{\|s\|_2}  = \|s\|_2 \leq 1,$$
	which is the conclusion we aimed for. In the equation above, we have used the following fact (see \cite[equation (5.4)]{passer2018minimal} for the corresponding statement): for non-negative scalars $a_1, \ldots, a_g$, 
	$$\left\|\sum_{i=1}^g a_i \overline{F_i} \otimes F_i \right\|_\infty = \sum_{i=1}^g a_i.$$
	The fact that the left hand side in the equality above is smaller than the right hand side follows from the triangle inequality. The reverse inequality follows from taking the scalar product against the maximally entangled state 
	$$\omega_d := \frac 1 d \sum_{i,j=1}^d (e_i \otimes e_i )(e_j \otimes e_j)^*,$$
	for some orthonormal basis $\{e_i\}_{i=1}^d$ of $\mathbb C^d$.
\end{proof}
Putting together the result above with Theorem \ref{thm:LB-Passer-et-al}, we derive the following equality, one of the main results of this paper. 
\begin{cor}\label{cor:large-d-Gamma-is-QC}
	For any $g \geq 2$ and any $d \geq 2^{\lceil (g-1)/2\rceil}$, we have 
	$$\Delta(g,d) = \Gamma(g,d) = \Delta^0(g,d) = \Gamma^0(g,d) = \mathrm{QC}_g.$$
\end{cor}

\begin{remark}
If the dimension bound $d \geq 2^{\lceil (g-1)/2\rceil}$ holds, the matrices $(F_1+I_d)/2 \oplus 0, \ldots, (F_g-I_d)/2 \oplus 0$ considered in this section are \emph{the most incompatible $g$-tuple of $d \times d$ quantum effects}. Indeed, for any direction $\hat s \in \mathrm{QC}_g$, $\|\hat s\| = 1$, it follows from Corollary \ref{cor:large-d-Gamma-is-QC} that the $g$-tuple $(t_1 \hat s_1 (F_1+I_d)/2, \ldots, t_g \hat s_g (F_g+I_d)/2)$ is compatible if and only if $t_i \leq 1$ for all $i \in [g]$. We would also like to point out that, for $d=2$ and $g=3$, $g=2$, the claim above corresponds to the maximal incompatibility of the measurements corresponding to the Pauli observables.
\end{remark}

\section{Discussion} \label{sec:discussion}

In this final section, we would like to put the results obtained in this work in perspective, and compare them with previously known bounds. We also list and discuss some questions left open in this work. 

\subsection{The shape of the different compatibility regions}
We start by listing some previously known results on the different sets $\Gamma^\#$ considered in this work. Let us remind the reader that our primary focus was on the sets $\Gamma(g,d)$, because of their connection to the inclusion problem for free spectrahedra. In the quantum information community, the sets $\Gamma^{all}$ and $\Gamma^{lin}$ play a very important role, because the most general type of trivial noise is allowed in the former case, and because of the linear structure in the latter case. Previously, mainly the cases of small $g,d$ have been considered in the literature. General lower bounds have been shown mostly using tools from symmetric approximate cloning, while upper bounds were rarely considered in the general case (we are considering here only the case of 2-outcome POVMs).

Let us first discuss the results in the literature for small $g,d$. Using an argument which connects joint measurability with a violation of the CHSH inequality \cite{Wolf2009measurements}, it was shown in \cite{busch2013comparing} that $\mathrm{QC}_2 \subseteq \Gamma(2,d)$ for all $d \in \mathbb N$. Further, it was shown that also $\Gamma^{all}(2,2) \subseteq \mathrm{QC}_2$ \cite[Proposition 3]{busch2008approximate} and $\Gamma^{0}(2,2) \subseteq \mathrm{QC}_2$
\cite[Proposition 4]{busch2008approximate}. Therefore, for $d \geq 2$, an application of Proposition \ref{prop:higherdimension} yields
\begin{equation*}
\Gamma(2,d) = \Gamma^{all}(2,d) = \Gamma^{0}(2,d) = \mathrm{QC}_2.
\end{equation*}
Less was known in the $g \geq 3$ case, since the connection to the CHSH inequality no longer holds \cite{Bene2018, Hirsch2018}.
From \cite{busch1986unsharp} (see also \cite[Section 14.4]{busch2016quantum}), it follows that $\mathrm{QC}_3 \subseteq \Gamma^0(3,2)$.
Moreover, \cite{Brougham2007} show that $\Gamma^0(3,2) \subseteq \mathrm{QC}_3$, hence $\Gamma^0(3,2) = \mathrm{QC}_3$. This was improved in \cite[Section XI]{Pal2011} to show also $\Gamma^{all}(3,2) \subseteq \mathrm{QC}_3$. Using the results of this paper and combining them with the findings above, we can prove a stronger statement.
An application of Theorem \ref{thm:LB-Passer-et-al} together with Proposition \ref{prop:higherdimension} yields
\begin{equation*}
\Gamma(3,d) = \Gamma^{all}(3,d)= \Gamma^{0}(3,d) = \mathrm{QC}_3.
\end{equation*}
In the general case, the lower bounds came mainly from symmetric cloning \cite{Heinosaari2014}, see Proposition \ref{prop:Gamma-clone-vs-Gamma-and-hat-Gamma}. 
 
 \bigskip
 
Let us now discuss the contributions of this paper to both the theory of joint measurability and free spectrahedra. As discussed in the introduction, our main insight, the relation between the joint measurability of 2-outcome POVMs and the inclusion problem for the matrix diamond, allows us to translate results from one field to the other. Arguably one of the main results in this work is the lower bound obtained in Theorem \ref{thm:LB-Passer-et-al}. Our theorem is based on results about inclusion of free spectrahedra derived in \cite{passer2018minimal}, which can be transferred to the quantum setting.
Together with the upper bound from \cite{passer2018minimal} and the lower bound from \cite{helton2014dilations},
 we obtain a much better understanding of the sets $\Gamma(g,d)$.
We present in Figure \ref{fig:picture-Gamma} our current picture of the sets $\Gamma(g,d)$, or, equivalently, of the sets of inclusion constants $\Delta(g,d)$. The curves $d = 2^{\lceil (g-1)/2 \rceil}$ and $d=\sqrt g/2$ delimit three regions: above the first curve, we know that the set $\Gamma(g,d)$ is equal to $\mathrm{QC}_g$, the positive part of the unit Euclidean ball, while below the second curve, we know the inclusion $\mathrm{QC}_g \subseteq\Gamma(g,d)$ to be strict.
Below the curve $d = 2^{\lceil (g-1)/2 \rceil}$, the upper bound from Theorem \ref{thm:UB-anti-commuting-F} does not apply, while below the second curve, $d=\sqrt g /2$, the lower bound $1/(2d)$ in the symmetric case is larger than the lower bound $1/\sqrt g$ coming from the quarter-circle $\mathrm{QC}_g$. It is worthwhile to mention that the best lower bound for the sets $\Gamma(g,d)$ coming from symmetric cloning (second line in Proposition \ref{prop:Gamma-clone-vs-Gamma-and-hat-Gamma}) is worse than the best of the two bounds coming from spectrahedron theory: 
$$\frac{g+2d}{g(1+2d)} \leq \max\left\{\frac{1}{\sqrt g}, \frac{1}{2d}\right\}.$$
However, in the asymmetric regime, cloning gives non-trivial lower bounds, since the $1/(2d)$ bound from Proposition \ref{prop:heltonsymmetric} is not applicable for asymmetric tuples. We expect to obtain non-trivial results as soon as $\frac{g+2d}{g(1+2d)}>1 /\sqrt g$, that is as soon as $g > 4d^2$. As an example, we plot in Figure \ref{fig:clone-vs-QC}, for even $g=2g_0$, the set of points of the form $(\underbrace{s, \ldots, s}_{g_0 \text{ times}},\underbrace{t, \ldots, t}_{g_0 \text{ times}})$ belonging to $\Gamma(g,4)$ and to $\mathrm{QC}_4$.

\begin{figure}[htbp]
	\begin{center}
		\includegraphics[scale=.5]{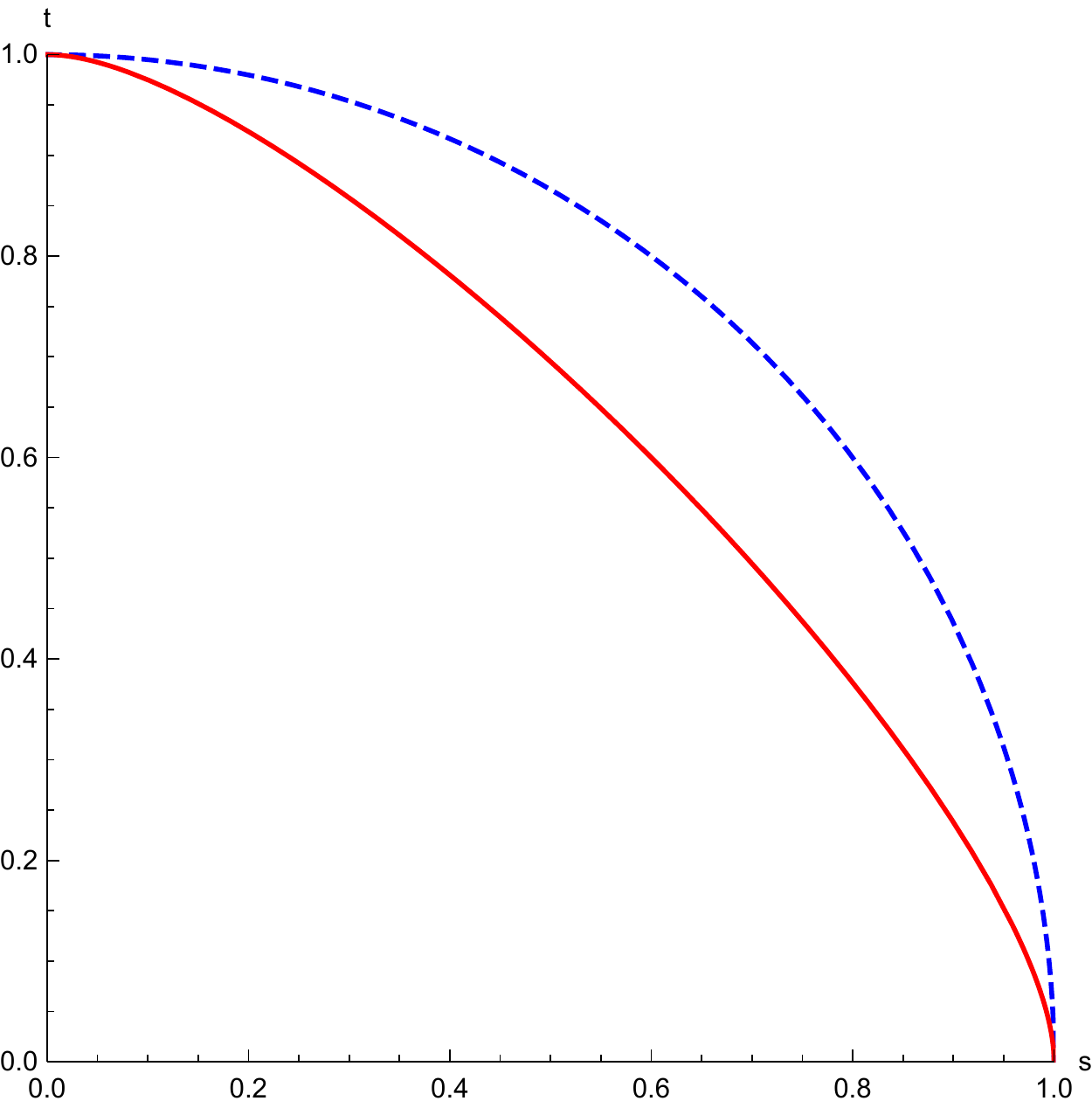} \qquad\qquad \includegraphics[scale=.5]{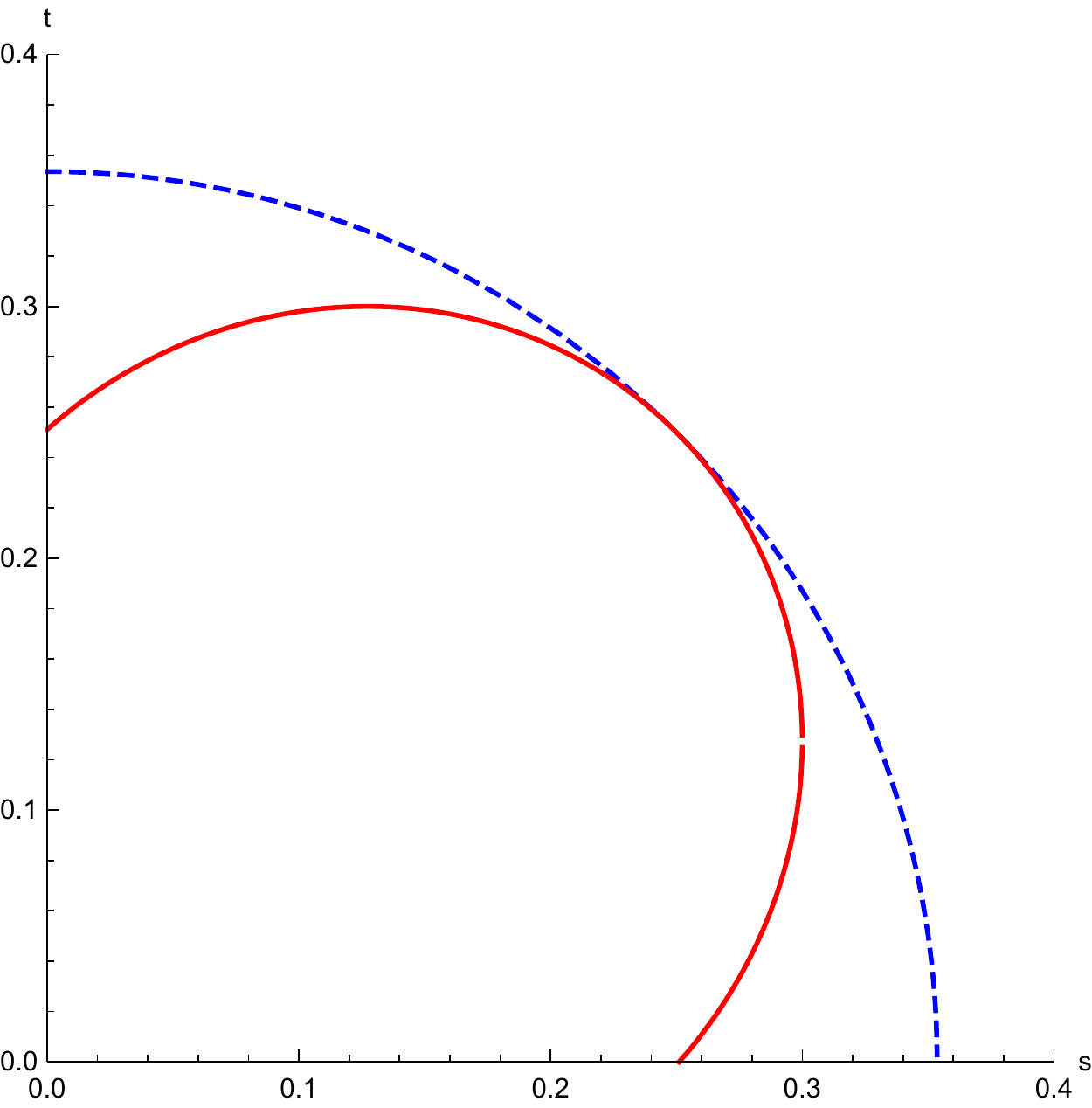} \\

	\includegraphics[scale=.5]{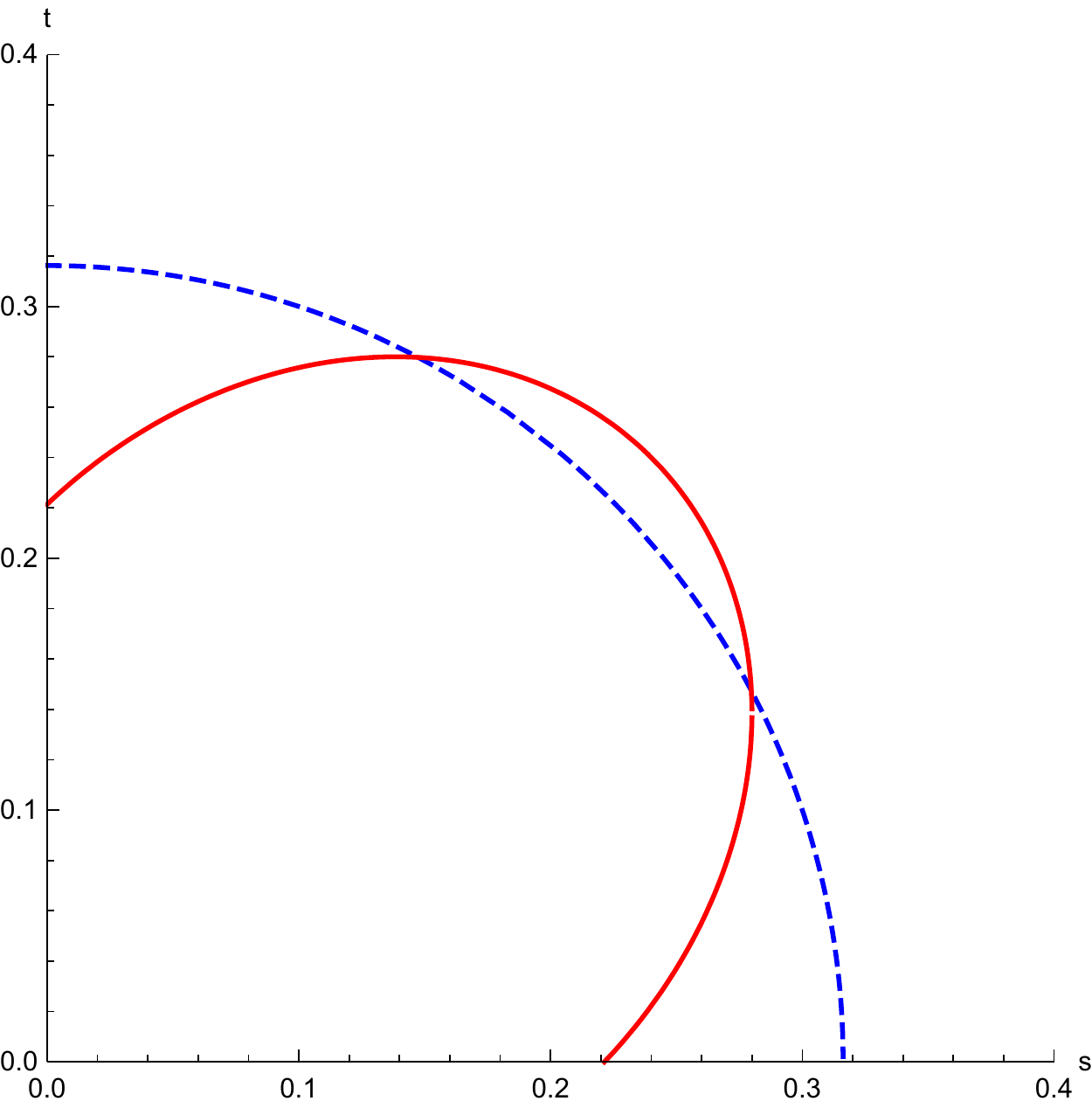} \qquad\qquad \includegraphics[scale=.5]{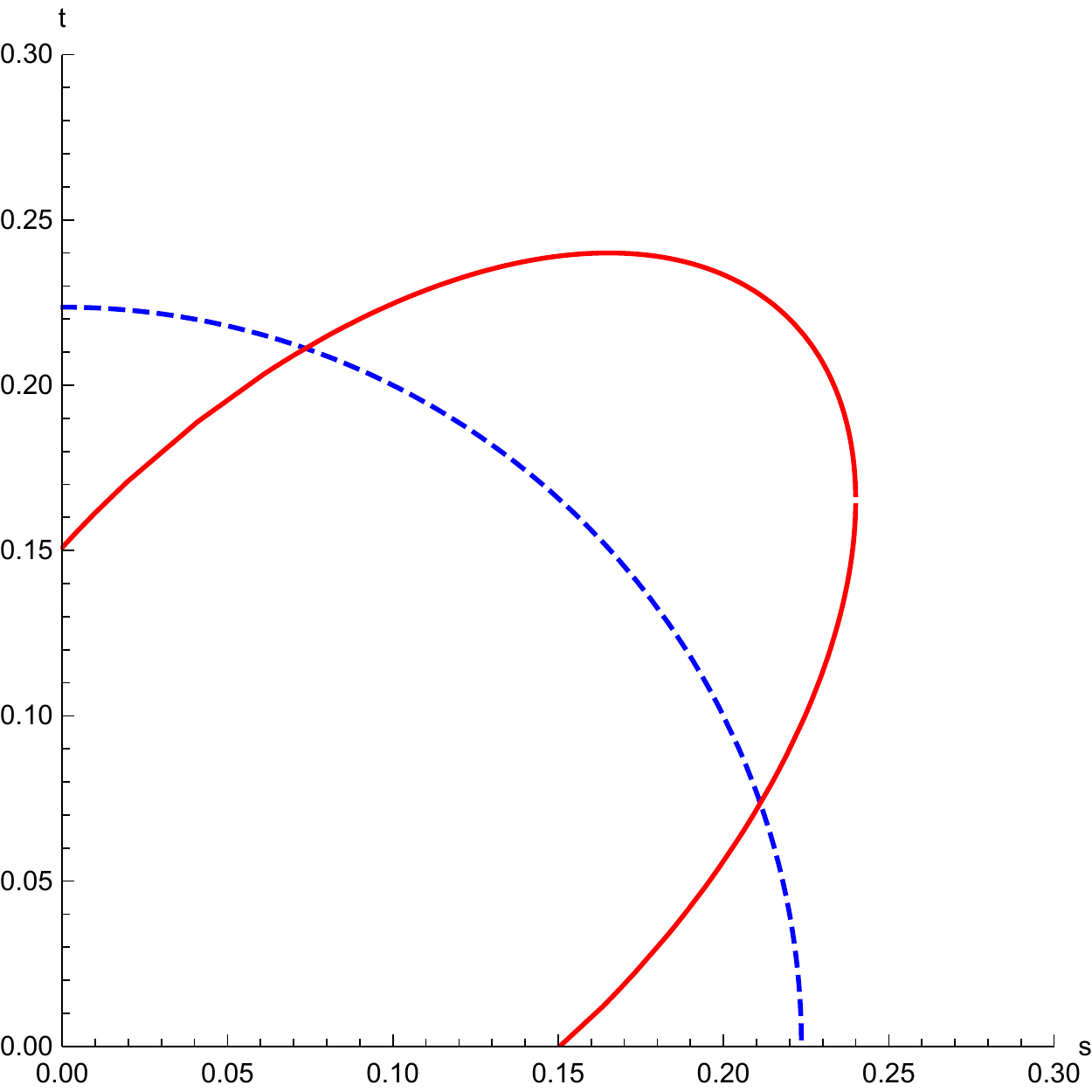} 
			\caption{The range of parameters $(s,t)$ for which the point $(s, \ldots, s, t, \ldots, t) \in [0,1]^g$, having the same number of $s$ and $t$, belongs to $\Gamma^{clone}(g,4)\subseteq \Gamma(g,2)$ (red curve) and $\mathrm{QG}_g$ (blue, dashed curve), for  $g=2,16,20,40$. Note that the bound from asymmetric cloning becomes	better for $g > 4d^2=16$.}
		\label{fig:clone-vs-QC}
	\end{center}
\end{figure}

The main question left open in this work is to compute the sets $\Gamma(g,d) = \Delta(g,d)$ for the range of parameters $g,d$ where the upper and lower bounds from \cite{passer2018minimal} do not agree.

\begin{question}
	Compute, for $d < 2^{\lceil (g-1)/2 \rceil}$, the sets $\Gamma(g,d) = \Delta(g,d)$.
	\end{question}

\begin{figure}[htbp]
	\begin{center}
		\includegraphics[scale=1]{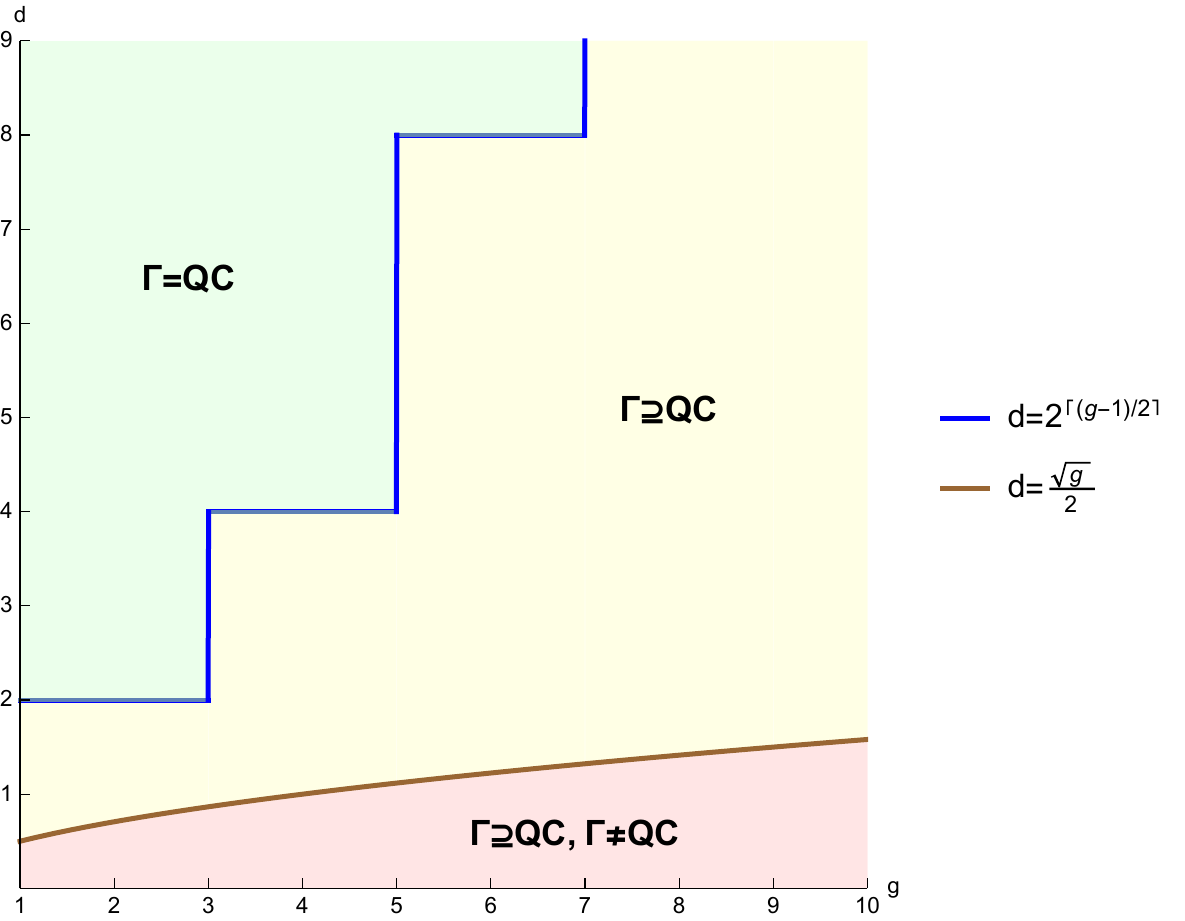}
		\caption{The sets $\Gamma(g,d)$.} 
		\label{fig:picture-Gamma}
	\end{center}
\end{figure}

Regarding the sets $\Gamma^{lin}(g,d)$, the lower bounds coming from cloning (see Section \ref{sec:cloning}) were already known in the literature; in particular, we recall that, in the symmetric case, 
$$\frac{g+d}{g(1+d)} (\underbrace{1, 1, \ldots, 1}_{g \text{ times}}) \subseteq \Gamma^{lin}(g,d).$$
Unfortunately, since there is no known inclusion of the $\Gamma$ sets into the $\Gamma^{lin}$ sets, we cannot use in this setting the very powerful lower bounds for $\Gamma(g,d)$ in Theorem \ref{thm:LB-Passer-et-al}; see also Remark \ref{rem:Gamma-subset-Gamma-lin}.

The upper bounds for the sets $\Gamma^{lin}(g,d)$ are new, and come from Zhu's criterion (Corollaries \ref{cor:Zhu-UB-Gamma-lin-MUB} and \ref{cor:Zhu-UB-Gamma-lin-SIC-POVM})
\begin{align*}
	\forall \text{$d$ prime power and }g \leq d+1, \qquad &\Gamma^{lin}(g,d) \subseteq \sqrt{d-1}\mathrm{QC}_g\\
	\forall \text{$d\in\{2,3,\ldots, 20, 23,27,29,30,\ldots\}$ and }g \leq (d-1)^2, \qquad &\Gamma^{lin}(g,d) \subseteq \sqrt{d-1}\mathrm{QC}_g
\end{align*}	
and from the work of Passer et al. (Theorem \ref{thm:UB-anti-commuting-F} and Proposition \ref{prop:inclusions-Gamma})
$$\forall d \geq 2^{\lceil (g-1)/2 \rceil},  \qquad \Gamma^{lin}(g,d) \subseteq\Gamma^{0}(g,d) \subseteq \mathrm{QC}_g.$$

Finally, regarding the sets $\Gamma^{all}$, allowing the most general type of noise, the new lower bounds obtained in this work are precisely the same as the ones for the sets $\Gamma$. Importantly, for all $g,d$, we have
	$$\mathrm{QC}_g \subseteq \Gamma^{all}(g,d).$$
Note that the bound above was previously known only in the case $g=2$. Moreover, in the symmetric case, we have 
	$$\max\left\{\frac{1}{\sqrt g}, \frac{1}{2d}\right\} (\underbrace{1, 1, \ldots, 1}_{g \text{ times}}) \subseteq \Gamma^{all}(g,d).$$
Upper bounds can be obtained via the map $F$ from Proposition \ref{prop:inclusions-Gamma} from upper bounds for $\Gamma$. For example, in the symmetric case, using Theorem \ref{thm:UB-anti-commuting-F} we get
$$\forall d \geq 2^{\lceil (g-1)/2 \rceil}, \qquad s (\underbrace{1, 1, \ldots, 1}_{g \text{ times}}) \subseteq \Gamma^{all}(g,d) \implies s \leq \frac{2}{1+\sqrt g}, $$
which is roughly two times the lower bounds above.

\subsection{The shape of the inclusion sets}

Hitherto, we have discussed the implications of this work for quantum information theory. However, our results also shed new light on $\Delta(g,d)$ and $\Delta^0(g,d)$, the sets of inclusion constants for the matrix diamond. 

As $\Delta^0(g,d) = \Gamma^0(g,d)$ by Theorem \ref{thm:jm=inclusion}, we have the lower bounds
\begin{align*}
\forall g,d \geq 1 \quad &\mathrm{QC}_g \subseteq \Gamma(g,d) \subseteq \Delta^0(g,d) \\
\forall g,d \geq 1 \quad& \Gamma^{clone}(g,d) \subseteq \Delta^0(g,d).
\end{align*}
Looking at the symmetric case, for which $s(1, \ldots, 1) \in \Gamma^{clone}$ if and only if $s \leq (g+d)/(g(1+d))$, we see that this is larger than $1/\sqrt{g}$ if and only if $d \leq \sqrt{g}$. Therefore, both lower bounds are non-trivial. We remark that for all $d$, $g \geq 1$, 
\begin{equation*}
\frac{1}{2d} \leq \frac{g+d}{g(1+d)}. 
\end{equation*}
Therefore, the result from symmetric cloning is always stronger than the one from \cite{helton2014dilations} (see Corollary \ref{cor:heltonspoint}). In terms of upper bounds, we only have that 
\begin{equation*}
\forall d \geq 2^{\lceil (g-1)/2 \rceil} \qquad \Delta^0(g,d) \subseteq \mathrm{QC}_g
\end{equation*}
from the work of Passer et al. (Theorem \ref{thm:UB-anti-commuting-F}).
 
Regarding the sets $\Delta(g,d)$ one obtains new lower bounds in the asymmetric setting, using cloning and the inclusion $\Gamma^{clone}(g,2d) \subseteq \Delta(g,d)$, see Figure \ref{fig:clone-vs-QC}.
\subsection{Outlook: POVMs with more outcomes}

In this work, we have focused on binary measurements. However, our methods also work for measurements with more outcomes. For example, consider the case of a binary POVM $\Set{E, I -E}$ and a POVM with three outcomes $\Set{F_1, F_2, I - F_1 - F_2}$. Then, it can be shown that joint measurability is equivalent to the inclusion problem of the free spectrahedra defined by
\begin{equation*}
A_1 = \frac{2}{3}\mathrm{diag}[2,-1,-1,2,-1,-1], ~ A_2 = \frac{2}{3}\mathrm{diag}[-1,2,-1,-1,2,-1], ~ A_3 =\mathrm{diag}[1,1,1,-1,-1,-1]
\end{equation*}
and
\begin{equation*}
B_1 = 2 E - I, \qquad B_2 = 2 F_1 - \frac{2}{3} I, \qquad B_3 = 2 F_2 - \frac{2}{3} I.
\end{equation*}
That is, $\mathcal D_A(1) \subseteq\mathcal D_B(1)$ if and only if $E, I-E, F_1, F_2$ and $I-F_1 - F_2$ are quantum effects and $\mathcal D_A \subseteq\mathcal D_B$ if and only if $\Set{E, I -E}$ and $\Set{F_1, F_2, I - F_1 - F_2}$ are jointly measurable POVMs. Inclusion constants then correspond again to mixing with $I/2$ (for the binary POVM) and $I/3$ (for the three-outcome POVM), respectively. This idea is explored in detail in \cite{bluhm2018jewel}.

\bigskip

\noindent {\it Acknowledgments.} A.B. would like to thank Michael M. Wolf for pointing out the connection between joint measurability and the CHSH inequality. A.B. acknowledges support from the ISAM Graduate Center at the Technische Universit\"at M\"unchen. I.N.~would like to thank Teiko Heinosaari for many inspiring discussions about joint measurability; in particular, he pointed out to us the fact in Remark \ref{rem:compression-jm-is-jm}. I.N.'s research has been supported by the ANR projects {StoQ} (grant number ANR-14-CE25-0003-01) and {NEXT} (grant number ANR-10-LABX-0037-NEXT), and by the PHC Sakura program (grant number 38615VA). I.N.~also acknowledges the hospitality of the Technische Universit\"at M\"unchen. A.B. and I.N. would like to thank the Institut Henri Poincar{\'e} in Paris for its hospitality and for hosting the trimester on ``Analysis in Quantum Information Theory'', during which part of this work was undertaken.

\bibliographystyle{alpha}
\bibliography{../spectralit}
\end{document}